%% file: arxiv-top.tex
\documentclass[a4paper,11pt]{article}

\usepackage{fullpage}
\usepackage{mathtools}
\usepackage{setspace}
\usepackage{amsmath}
\usepackage{amsthm}
\usepackage{amssymb}
\usepackage{titlesec}
\usepackage{xspace}
\usepackage{accents}
\usepackage{hyperref}

\newtheorem{thm}{Theorem}[section]
\newtheorem{theorem}[thm]{Theorem}
\newtheorem{claim*}{Claim}

\newtheorem{definition}[thm]{Definition}

\newtheorem{lemma}[thm]{Lemma}
\newtheorem{proposition}[thm]{Proposition}

\title{Separating LREC from LFP\thanks{Research funded in part by EPSRC grant EP/S03238X/1.}}

\author{Anuj Dawar and Felipe Ferreira Santos \\
  Department of Computer Science and Technology\\ University of Cambridge.\\
  \texttt{anuj.dawar@cl.cam.ac.uk, ff334@cam.ac.uk}}

\begin{document}

\input{macros.tex}

\maketitle

\begin{abstract}
$\LREC_=$ is an extension of first-order logic with a logarithmic recursion operator.  It was introduced by Grohe et al.\ and shown to capture the complexity class $\LS$ over trees and interval graphs.  It does not capture $\LS$ in general as it is contained in $\FPC$---fixed-point logic with counting.  We show that this containment is strict.  In particular, we show that the path systems problem, a classic $\PT$-complete problem which is definable in $\LFP$---fixed-point logic---is not definable in $\LREC_=$.  This shows that the logarithmic recursion mechanism is provably weaker than general least fixed points.  The proof is based on a novel Spoiler-Duplicator game tailored for this logic.
\end{abstract}

\section{Introduction}\label{sec:intro}

\input{intro.tex}

\section{Preliminaries}\label{sec:prelims}
\input{preliminaries.tex}

\section{\texorpdfstring{\textsf{LREC}$_=$}{Lg} Game}\label{game-sec}

\input{game.tex}

\section{Duplicator Winning Pebblings and Extender functions}\label{sec:extender}

\input{extender.tex}

\section{Inexpressibility of PSP in \texorpdfstring{\textsf{LREC}$_=$}{Lg}}
\label{sec:result}

\input{main-theorem.tex}

\section{Conclusion}
\input{conclusion.tex}

\bibliographystyle{plainurl}
\bibliography{refs}

\end{document}

%% file: macros.tex
\newcommand{\PT}{\ensuremath{\mathrm{P}}\xspace}
\newcommand{\NP}{\ensuremath{\mathrm{NP}}\xspace}
\newcommand{\LS}{\ensuremath{\mathrm{L}}\xspace}
\newcommand{\PSP}{\ensuremath{\mathrm{PSP}}\xspace}

\newcommand{\logic}[1]{\textsf{\upshape #1}\xspace} 

\newcommand{\FPC}{\logic{FPC}}
\newcommand{\FOC}{\logic{FOC}}
\newcommand{\FO}{\logic{FO}}
\newcommand{\LFP}{\logic{LFP}}
\newcommand{\LREC}{\logic{LREC}}
\newcommand{\lrec}{\logic{lrec}}
\newcommand{\DTC}{\logic{DTC}}
\newcommand{\CL}{\logic{CLog}}

\newcommand{\Acal}{\mathcal{A}}
\newcommand{\Bcal}{\mathcal{B}}
\newcommand{\Ccal}{\mathcal{C}}
\newcommand{\Dcal}{\mathcal{D}}
\newcommand{\Gcal}{\mathcal{G}}
\newcommand{\Hcal}{\mathcal{H}}
\newcommand{\Ical}{\mathcal{I}}
\newcommand{\Kcal}{\mathcal{K}}
\newcommand{\Lcal}{\mathcal{L}}
\newcommand{\Ncal}{\mathcal{N}}
\newcommand{\Pcal}{\mathcal{P}}
\newcommand{\Vcal}{\mathcal{V}}
\newcommand{\Zcal}{\mathcal{Z}}

\newcommand{\Afrak}{\mathfrak{A}}
\newcommand{\Bfrak}{\mathfrak{B}}
\newcommand{\Cfrak}{\mathfrak{C}}
\newcommand{\Dfrak}{\mathfrak{D}}
\newcommand{\Sfrak}{\mathfrak{S}}

\newcommand{\ctt}{\mathtt{c}}
\newcommand{\Ctt}{\mathtt{C}}
\newcommand{\Ett}{\mathtt{E}}
\newcommand{\Gtt}{\mathtt{G}}
\newcommand{\Utt}{\mathtt{U}}
\newcommand{\Vtt}{\mathtt{V}}
\newcommand{\Ytt}{\mathtt{Y}}
\newcommand{\Xtt}{\mathtt{X}}

\newcommand{\ROOT}{\mathtt{\underline{root}}}

\newcommand{\src}{S}
\newcommand{\tgt}{\mathtt{t}}

\newcommand{\gf}[1]{\mathbb{Z}_{#1}}
\newcommand{\id}[1]{\mathcal{I}_{#1}}
\newcommand{\trans}[1]{\widehat{#1}}
\newcommand{\outedge}[2]{[#1]_\mathfrak{#2}\mathtt{E}_\mathfrak{#2}}
\newcommand{\trc}{\mathrel{\accentset{\ast}{\sim}}}
\newcommand{\nst}[2]{\operatorname{norm}(#1,#2)}
\newcommand{\nnm}[2]{\|#1\|_{#2}}
\newcommand{\symm}[1]{\operatorname{Sym}(#1)}
\newcommand{\rk}[1]{\operatorname{rank}(#1)}
\newcommand{\dg}[1]{\operatorname{deg}(#1)}

\newcommand{\dom}{\operatorname{dom}}
\newcommand{\range}{\operatorname{range}}
\newcommand{\lift}[2]{\operatorname{lift}_{#1}(#2)}
\newcommand{\proj}{\operatorname{proj}}
\newcommand{\bij}{\operatorname{bij}}

\newcommand{\cl}[1]{\overline{#1}}
\newcommand{\mnh}[1]{\ensuremath{\text{min-h}(#1)}}
\newcommand{\mxh}[1]{\ensuremath{\text{max-h}(#1)}}
\newcommand{\hgt}[1]{\ensuremath{\text{height}(#1)}}
\newcommand{\nh}[1]{\ensuremath{\text{null-h}(#1)}}
\newcommand{\ih}[1]{\ensuremath{\text{iso-h}(#1)}}

\newcommand{\ra}{\rightarrow}
\newcommand{\ZZ}{\mathbb{Z}}
\newcommand{\Zp}{\ZZ/p\ZZ}
\newcommand{\voc}{\vartheta}
\newcommand{\nats}{\mathbb{N}}

\renewcommand{\phi}{\varphi}

%% file: intro.tex
The founding result of the subject of descriptive complexity is Fagin's characterization of $\NP$ as the class of properties definable in existential second-order logic~\cite{fagin}.  Similar characterizations for complexity classes below $\NP$ are not known and remain an active area of investigation.  Much attention has been devoted to the question of whether there is a logic for $\PT$, and this is often said to be the central open question of descriptive complexity.  Perhaps less well known but still wide open is the question of whether there is a logic giving a descriptive characterization of $\LS$---the class of properties of finite structures decidable deterministically in logarithmic space.

Immerman established logical characterizations of $\PT$ and $\LS$ on \emph{ordered} structures.  The former is given by the properties definable in $\LFP$---the extension of first-order logic with a least fixed-point operator (a result obtained independently by Vardi)---and the latter by $\DTC$---the extension of first-order logic with a deterministic transitive closure operator.   These logics are known to be strictly weaker than the corresponding complexity classes in the absence of a built-in order (see~\cite{Imm99} for details).  Moreover, on unordered structures, it is known that the expressive power of $\DTC$ is strictly weaker than that of $\LFP$~\cite{GraedelM95}, whereas the corresponding question for ordered structures is equivalent to the separation of $\LS$ from $\PT$.

Since the simplest properties separating $\LFP$ and $\DTC$ from the complexity classes $\PT$ and $\LS$ respectively are counting properties, there has also been much interest in the extensions of these logics with counting mechanisms.  The logic $\FPC$---fixed-point logic with counting---is widely studied (see~\cite{Dawar15} for a brief introduction).  Though it is known to be strictly weaker than $\PT$ by a celebrated result of Cai, F\"urer and Immerman~\cite{CFI92}, it gives a robust definition of the class of problems in $\PT$ which are symmetrically solvable~\cite{ADO19} and captures $\PT$ on a wide range of graph classes.  Most significantly, Grohe has proved that any polynomial-time property of graphs that excludes a fixed graph as a minor is definable in $\FPC$~\cite{Grohe17}.  Results showing that $\FPC$ captures $\PT$ on a class of structures $\Ccal$ (including Grohe's theorem) are typically established by showing that a polynomial-time canonical labelling algorithm on the class $\Ccal$ can be implemented in $\FPC$.  Such results were established for specific classes of graphs such as trees~\cite{ImmermanL90} and graphs of bounded tree-width~\cite{GroheM99} before Grohe's theorem which supersedes them all.

The situation for logics capturing $\LS$ is less clear-cut, even for classes of structures where it is known that logarithmic-space canonical labelling algorithms are possible.  For instance, Etessami and Immerman~\cite{etessami} showed that the extension of $\DTC$ with counting fails to capture $\LS$ even on trees.  What this suggests is that the weakness of $\DTC$ is not just in the lack of a means of counting but that the recursion mechanism embodied in the deterministic transitive closure operator is too weak.  An interesting suggestion to remedy this is in the logic $\LREC$ introduced in~\cite{GroheGHL12}, which incorporates a rich recursion mechanism which can still be evaluated within logarithmic space.  To be precise, the paper introduces two versions of the logic, one called $\LREC$ and the second $\LREC_{=}$.  While the first is shown to be sufficient to capture the complexity class $\LS$ on trees, it is the latter that is the more robust logic.  In particular, $\LREC_{=}$ properly extends $\LREC$ in expressive power, is closed under first-order interpretations and (unlike $\LREC$) can express undirected reachability.  Moreover, $\LREC_{=}$ has also been shown to capture $\LS$ on other interesting classes of structures~\cite{Grussien19}.  

It is known that $\LREC_{=}$ does not capture $\LS$ on all graphs.  In particular, $\LREC_{=}$ is included in $\FPC$ and the CFI construction that separates $\FPC$ from $\PT$ can be suitably padded to give a separation of $\LREC_{=}$ from $\LS$~\cite{clog}.  While Grohe et al.\ proved that $\LREC$ is properly contained in $\FPC$ (for instance, it cannot express reachability on undirected graphs), the question was left open for the stronger logic $\LREC_{=}$.  We settle this question in the present paper.  That is, we prove that the expressive power of $\LREC_{=}$ is strictly weaker than $\FPC$.  Indeed, we show that the \emph{path systems problem} (PSP), a natural $\PT$-complete problem that is expressible in $\LFP$ (even without counting) is not expressible in $\LREC_{=}$.  The result is established by first defining an Ehrenfeucht-style pebble game for the $L$-recursion quantifier used in the logic $\LREC_{=}$ and then deploying it on an intricate class of instances of PSP.  Describing the winning Duplicator strategy in the game is challenging and takes up the bulk of the paper.

%% file: preliminaries.tex
We assume that the reader is familiar with the definitions of the
basic logics used in finite model theory, in particular first-order
logic $\FO$, fixed-point logic $\LFP$ and their extension swith
coungint, $\FOC$ and $\FPC$ respectively.  These definitions can be found in standard
textbooks~\cite{EF99,libkin}.  We consider vocabularies with relation
and constant symbols, which we call \emph{relational vocabularies} for
short.  Formulas in a vocabulary $\sigma$ are interepreted in finite
$\sigma$-structures.  We often use Fraktur letters
$\Afrak,\Bfrak,\ldots$ for structures and the corresponding Roman
letter $A,B,\ldots$ to denote the universe of the structure.

In the case of logics with counting, such as $\FOC$ and 
$\FPC$, the interpretation also involves a \emph{number domain}.  The
logic allows for two kinds of variables (and more generally, two kinds
of terms) which range over elements and numbers respectively.  Thus, 
we consider  a $\sigma$-structure $\Afrak$ to be extended with a
domain of non-negative integers up to the cardinality of $\Afrak$.
Element variables are interpreted by elements of $\Afrak$ and number
variables by elements of the number domain.  We write $\Ncal(A)$ to
denote the number domain associated with the structure $\Afrak$.  This
can be identified with the set $\{0,\ldots,|A|\}$ and we assume this
is disjoint from the set $A$.  The logics include a counting operator
and we assume a normal form where every occurrence of the counting
operator is in a formula of the form $[\# x \phi] = \nu$ where $\nu$
is a number variable.  This formula is to be read as saying that the
number of elements $x$ satisfying $\phi$ is $\nu$.

Any formula in any of the logics defines a \emph{query}.  
If $\phi$ is a formula in vocabulary $\sigma$ and $\vec{x}$ is an $m$-tuple of element variables such
that all free variables of $\phi$ are among $\vec{x}$, then $\phi(x)$
defines an \emph{$m$-ary query}.  This is a map taking each
$\sigma$-structure $\Afrak$ to an $m$-ary relation $\phi^{\Afrak}$ on
the universe of $\Afrak$ containing the tuples $\vec{a}$ such that
$\Afrak \models \phi[\vec{a}]$.  If $m=0$ and so $\phi$ is a sentence,
we call the query a \emph{Boolean query}.  A Boolean query is often
identified with a class of $\sigma$-structures.  Note that the query is not completely specified by the formula $\phi$ alone, but rather by the formula along with the tuple of variables $\vec{x}$.  In particular the order on the variables is given by the tuple.

To allow for formulas which may also have free number variables, we consider  \emph{numerical queries}.  An $m$-ary numerical query is a map $Q$ that takes any $\sigma$-structures $\Afrak$ to a relation $Q(\Afrak) \subseteq A^m \times \nats$ and is invariant under isomorphisms.  The last condition means that if $\eta: A \ra B$ is an isomorphism from $\Afrak$ to $\Bfrak$, then $(\vec{a},n) \in Q(\Afrak)$ if, and only if, $(\eta(\vec{a}),n) \in Q(\Bfrak)$.  Such queries arise quite naturally.  As an example, the function that takes a graph $G$ to the number of connected components in $G$ is a $0$-ary numerical query.  The function taking a graph $G$ to the set of pairs $(a,n)$ where $a$ is a vertex in $G$ and $n$ the number of vertices in the connected component of $a$ is a $1$-ary numerical query.  

Suppose $\phi$ is a formula of a logic with counting, such as $\FOC$ or $\FPC$, with free variables among $\vec{x}$, where $\vec{x}$ contains $m_1$ element variables and $m_2$ number variables.  Without loss of generality we assume all element variables appear in $\vec{x}$ before the number variables.   Interpreted in a structure $\Afrak$, the formula $\phi$ defines a relation of mixed type: a subset of $A^{m_1} \times \Ncal(A)^{m_2}$.  We treat the tuples of numbers as coding positive integers in a natural way.  So, suppose $n = |A|$ and therefore $\Ncal(A) = \{0,\ldots,n\}$.  For a tuple $\vec{r} = (r_1,\ldots,r_{m_2})$ we write $\langle \vec{r} \rangle$ denote the number $$\sum\limits_{i=1}^{m_2} (n+1)^{i-1}r_i.$$
In this way, we can see $\phi$ as defining an $m_1$-ary numerical query which takes $\Afrak$ to the set $\{ (\vec{a},\langle \vec{r} \rangle) \mid \Afrak \models \phi[\vec{a},\vec{r}] \}$.

\subsection{Generalized Quantifiers and Operators}\label{sec:genop}
In the next subsection we define the logics $\LREC$ and $\LREC_=$
which were introduced by Grohe et al.~\cite{Grohe_2013}.  They were
originally defined as an extension of first-order logic with a new
kind of recursion operator called $L$-recursion, designed to be
computable in logspace.  We present an equivalent formulation through
a variation of generalized quantifiers in the style of Lindstr\"om.
To do this, we first briefly review interpretations and generalized quantifiers.  In what follows, fix a logic $\Lcal$, which can be any of the logics $\FO$, $\FOC$, $\LFP$ or  $\FPC$.  Whenever we refer to formulas, this is to be read as formulas of $\Lcal$.

Given two relational vocabularies $\sigma$ and $\tau$ and a positive
integer $d$, where $\tau = (R_1,\ldots,R_s)$ contains no constant
symbols, an interpretation of $\tau$ in $\sigma$ of dimension $d$ is a
sequence $\Ical$ of formulas 
$$(\delta(\vec{x},\vec{w}),\varepsilon(\vec{x},\vec{y},\vec{w}),\phi_{R_1}(\vec{z}_1,\vec{w}),\ldots,\phi_{R_s}(\vec{z}_s,\vec{w})),$$
where $\vec{x}$ and $\vec{y}$ are tuples of variables of length $d$,
and each $\vec{z}_i$ is a tuple of variables of length $dr_i$ where
$r_i$ is the arity of the relation symbol $R_i$.  The variables $\vec{w}$ are the \emph{parameter variables}.  Note that all tuples of variables may contain both element variables and number variables.

Let $\Afrak$ be a $\sigma$ structure and $\vec{p}$ be an interpretation for the parameter variables $\vec{w}$.  The interpretation~$\Ical$ associates a~$\tau$-structure~$\Bfrak$ to
$\Afrak,\vec{p}$ if there is a map~$h$ from~$\{ \vec{a} \in
(A \cup \Ncal(A))^d \mid \Afrak \models \delta[\vec{a},\vec{p}] \}$ to the universe~$B$ of~$\Bfrak$
such that: (i)~$h$ is surjective onto~$B$; (ii)~$h(\vec{a}_1) =
h(\vec{a}_2)$ if, and only if,~$\Afrak\models \varepsilon[\vec{a}_1,
  \vec{a}_2,\vec{p}]$; and (iii)~$R^{\Bfrak}(h(\vec{a}_1), \dots, h(\vec{a}_k))$
if, and only if,~$\Afrak \models \phi_R[\vec{a}_1, \dots, \vec{a}_k,\vec{p}]$.  Note that an interpretation~$\Ical$ associates
a~$\tau$-structure with~$\Afrak,\vec{p}$ only if~$\varepsilon$ defines an
equivalence relation on~$(A \cup \Ncal(A))^d$ that is a congruence with respect to the
relations defined by the formulaa~$\phi_R$. In such
cases, however,~$\Bfrak$ is uniquely defined up to isomorphism and we
write~$\Ical(\Afrak,\vec{p}) = \Bfrak$.

There are many kinds of interpretation that are defined in the
literature.  The ones we have defined here are fairly generous, in
that they allow for \emph{relativization}, \emph{vectorization} and
\emph{quotienting}.  The first of these means that the universe of
$\Ical(\Afrak)$ is defined as a subset of $(A \cup \Ncal(A))^d$ rather
than the whole set.  The formula $\delta$ is the relativizing
formula.  Vecotorization refers to the fact that $d$ can be greater
than $1$ and so $\Ical$ can map $\sigma$-structures to structures that
are polynomially larger, and quotienting refers to the fact that the
universe is obtained by quotienting with the congruence relation
defined by $\varepsilon$.

Let $K$ be an isomorphism-closed class of $\tau$-structures (or equivalently a
Boolean query).  A standard way of extending any logic $\Lcal$ to obtain a minimal extension $\Lcal[K]$ in which the Boolean query $K$ can be expressed  is to adjoin to $\Lcal$ the Lindstr\"om quantifier $Q_K$, corresponding to $K$ (see~\cite[Chapt.~12]{EF99}).  This allows us to write formulas of the form $Q_K\Ical$ where $\Ical$ is an interpretation as above.  The quantifier $Q_K$ binds the variables $\vec{x},\vec{y},\vec{z}_1,\ldots,\vec{z}_s$ appearing in $\Ical$ so the free varialbles of $Q_K\Ical$ are $\vec{w}$.  The formula $Q_K\Ical$ is true in a $\sigma$-structure $\Afrak$ with an intepretation $\vec{p}$ for the free variables if, and only if, $\Ical(\Afrak,\vec{p})$ is in $K$.

We generalize such quantifiers to queries that are not necessarily Boolean.  In general, consider an $m$-ary numerical query $K$ over $\tau$-structures.  We define $\Lcal[K]$, the extension of the logic $\Lcal$ with an operator for $K$ as follows.  Let $\Ical = (\delta(\vec{x},\vec{w}),\varepsilon(\vec{x},\vec{y},\vec{w}),\phi_{R_1}(\vec{z}_1,\vec{w}),\ldots,\phi_{R_s}(\vec{z}_s,\vec{w}))$ be an interpretation of $\tau$ in $\sigma$ of dimension $d$.  Let $\vec{u}$ be a tuple of variables of length $md$ and $\vec{s}$ a tuple of number terms.  Then,
$$[{Q_K}_{\vec{x},\vec{y},\vec{z}_1,\ldots,\vec{z}_s}\Ical](\vec{u},\vec{s})$$
is a formula with free variables $\vec{w},\vec{u},\vec{s}$.  This formula is true in a $\sigma$-structure $\Afrak$, with an interpretation of $\vec{p}$ for $\vec{w}$, $\vec{a}$ for $\vec{u}$ and $\vec{r}$ for $\vec{s}$ if, and only if,
$$([\vec{a}]_{\varepsilon},\langle \vec{r} \rangle) \in K(\Ical(\Afrak,\vec{p})).$$
Here, $[\vec{a}]_{\varepsilon}$ denotes the $m$-tuple of elements of $\Ical(\Afrak,\vec{p})$ obtained by taking the equivalence classes of the $m$ $d$-tuples that make up $\vec{a}$ under the equivalence relation defined by $\varepsilon$.

\subsection{\texorpdfstring{\LREC$_=$}{Lg}}\label{sec:logic-def}
In this subsection we define $\LREC_=$, the logic in which we prove the inexpressibility result.  We define the logic as an extension of $\FOC$ by means of a generalized operator in the sense of Section~\ref{sec:genop} above.  This definition is superficially different from that of Grohe et al.~\cite{GroheGHL12} where it is defined as a restricted fixed-point operator, but the two definitions are easily seen to be equivalent.

We first define a unary numeric query $\chi$ on a class of \emph{labelled}
graphs.  These are directed graphs $G = (V,E)$ together with a
labelling $C$ that gives for each vertex $v \in V$ a set of numbers
$C(v) \subseteq \Ncal(V)$.
For each such $(G,C)$ we  define
$\chi(G,C) \subseteq V \times \mathbb{N}$.  To do this, we first
introduce some notation.  For any  binary relation $R\subseteq A^2$,
and $a \in A$, let $aR$ denote the set $\{b\in A \mid (a,b)\in R\}$
and  $Ra$ denote  $\{b\in A \mid (b,a)\in R\}$.  Then $\chi$ is
defined by recursion on $l$ by the following rule:
    \[ (u, \ell) \in \chi(G,C) :\iff \ell \geq 0 \text{ and } |\{v \in uE \mid (v, \lfloor (\ell - 1)/|Ev|\rfloor) \in \chi(G,C)\}| \in C(u). \]

The logic $\LREC$ is the closure of $\FOC$ with a generalized operator
for computing the numeric query $\chi$.  Note that to get this
formally, we have to extend the idea of interpretation we had earlier
to be able to produce a \emph{labelled} graph.  In a logic with
couting, this is easily achieved as the interpretation will allow for
formulas defining relations of mixed element and number type.

To define the logic $\LREC_=$, we considered labelled
\emph{semi-graphs}. 
A semi-graph is a structure $G = (V,E,\sim)$, where $V$ is a set of
nodes,  and $E$ and $\sim$ are two binary relations on $V$.  Again, we
consider a semi-graph together with a labelling $C$ such that for each
$v \in V$, $C(v) \subseteq \Ncal(V)$.
 Let $\cong$ be the symmetric reflexive transitive closure of $\sim$
 and $[v]$ denote the equivalence class of $v$ under $\cong$.  Write
 $V/_{\cong}$ for set of such equivalence classes, i.e.\ the quotient
 of $V$ under this equivalence relation.  We write $[G]$ for
 the directed graph with vertex set $V/_{\cong}$ and edge set
 $$ [E] := \{ ([a],[b]) \mid (a',b') \in E \text{ for some } a' \in
 [a] \text{ and } b' \in [b] \}.$$
Write $\hat{C}$ for the labelling of $[G]$ given by $\hat{C}([v]) =
\bigcup_{u \in [v]}C(u)$.  We can then define the numeric query
$\hat{\chi}$ on labelled semi-graphs by the condition:
$$\hat{\chi}(G,C) = \chi([G],\hat{C}).$$

The logic $\LREC_=$ is then defined as the closure of $\FOC$ under a
generalized operator for the numeric query $\hat{\chi}$.  For
simplicity, we only allow applications of the operator to
interpretations without relativization of quotienting.  We explian below
why this is no loss of generality.  Indeed, it is also easily checked
that this restricted definition corresponds exactly to the original definition
of the logic $\LREC_=$ given by Grohe et al.

The syntax of the logic can then be defined as follows.
\begin{definition}\label{def:lrec-syntax}
  For a vocabulary $\tau$, the set of formulas of  $\LREC_=[\tau]$ is defined by extending $\FOC[\tau]$ formula formation rules by the following rule.
Suppose $\varphi_\sim(\vec{u},\vec{v}),\varphi_F(\vec{u},\vec{v})$, and $\varphi_C(\vec{u},\vec{p})$ are $\LREC_=$ formulas where $\vec{u}$,$\vec{v}$ and $\vec{w}$ are tuples of variables all of length $c$; $\vec{p}$ is a tuple of \emph{number} variables of length $d \leq c$ and  $\vec{r}$ is a tuple of number variables of length $q$, then 
 $$\varphi := [\textsf{lrec}_{\vec{u},\vec{v},\vec{p}}\varphi_E,\varphi_\sim,\varphi_C](\vec{w},\vec{r})$$ is an $\LREC_=$ formula.  The free variables of the formula are those in $\vec{w}$ and $\vec{r}$, along with the free variables in $\varphi_E$ and $\varphi_\sim$ that are not among $\vec{u}$ and $\vec{v}$ and those of $\varphi_C$ that are not among $\vec{u}$ and $\vec{p}$.
\end{definition}

To define the semantics of the logic, consider the formula $\textsf{lrec}_{\vec{u},\vec{v},\vec{p}}\varphi_E,\varphi_\sim,\varphi_C](\vec{w},\vec{r})$ interpreted in a $\tau$-structure $\Afrak$.  Let the semigraph $G = (A^c,E,\sim)$ be given by the interpretation $\varphi_E, \varphi_\sim$ on $\Afrak$ and $C$ be the labelling defined by $\varphi_C$.  Then the formula is satisfied in $\Afrak$ with $\vec{a}$ and $\vec{s}$ interpreting $\vec{w}$ and $\vec{r}$ respectively if, and only if, $(\vec{a},\langle \vec{s} \rangle) \in \hat{\chi}(G,C)$.

Note that this amounts to extending $\FOC$ with a generalized operator for the numeric query $\hat{\chi}$.  We have not allowed for relatavization or quotienting in the interpretations in our definition and we now explain why this involves no loss of generality.  The definition of $\hat{\chi}$ over the semi-graph $G$ is defined by taking $\chi$ over the quotient of $G$ with respect to the reflexive, symmetric and transitive closure of $\sim$.  Thus, we can replace a definable congruence in an interpretation by combining it with the definition of $\sim$ without changing the semantics.  By the same token, we can replace any relativizing formula $\delta$ by incorporating it in the definition of the binary relation $E$ and ensuring that only tuples that satisfy $\delta$ are involved in this binary relation.  It is easily seen from the definition of $\chi$ that for any vertex $v$ that has no $E$-neighbours, either $0$ is in $C(v)$ and thus all pairs $(v,\ell)$ are in $\chi(G,C)$ for all $\ell \in \mathbb{N}$, or 0 is not in $C(v)$ and thus no such pairs $(v,\ell)$ are in $\chi(G,C)$.

We define two parameters by which we measure the size of an $\LREC_=$
formula, rank and iteration-degree. As a base case, the rank and
iteration-degree of any atomic formula $\phi$ is 0.  If $\phi$ has
rank $r$ and iteration-degree $d$, then $\exists x \phi$, $\forall
x \phi$ and $[\# x \phi] = \nu$ all have rank $r+1$ and iteration degree $d$  Finally,  if we have the formula
$$\varphi := [\textsf{lrec}_{\vec{u},\vec{v},\vec{p}}\varphi_E,\varphi_\sim,\varphi_C](\vec{w},\vec{r})$$
where the length of $\vec{u},\vec{v}$ is $c$ and $\vec{p}$ is $d$, then the rank of $\varphi$, denoted $\rk{\varphi}$, is equal to $$\max[2c + \rk{\varphi_\sim}, 2c + \rk{\varphi_E}, c + d + \rk{\varphi_C}]$$ and if $\vec{r}$ has length $q$ then the iteration-degree of $\varphi$, denoted $\dg{\varphi}$, is equal to $$\max[q,\dg{\varphi_\sim},\dg{\varphi_E}, \dg{\varphi_C}].$$

Note that the rank and iteration-degree values may be distinct from each other, and the iteration-degree need not be larger than that of sub-formulas, regardless of the form of $\varphi$. This will be important when we introduce the game for the logic.

\subsection{The Path Systems Problem}
In this subsection we define the path systems problem (\PSP). Moreover we will define a specific subclass of \PSP instances which we will be using to prove that \PSP is inexpressible in $\LREC_=$. 

The problem is defined as a class of structures in a vocabulary with  a ternary relation $R$, unary relation $\src$, and constant $\tgt$.  Given a universe $U$ and a relation $R \subseteq U^3$, we define the upward closure of any set $X \subseteq U$ as the smallest superset $Y$ of $X$ such that for $a,b,c\in U$, if $a,b$ is in $Y$ and $(a,b,c)$ is in $R$, then $c$ is in $Y$.

The decision problem $\PSP$ then consists of those structures $\Afrak$ for which $\tgt^\Afrak$ is in the upwards closure of $\src^\Afrak$.


As we will be using \PSP to prove a separation between $\LREC_=$ and $\LFP$, we verify that \PSP is definable in $\LFP$.  This is well known.  For completeness, we give the defining sentence.
$$\varphi := [\textsf{lfp}_{X,u}\src(u)\vee(\exists v,w. X(v) \wedge X(w) \wedge R(v,w,u))](\tgt).$$

%% file: game.tex
In order to prove the inexpressibility result for $\LREC_=$, we
introduce a Spoiler-Duplicator game.  This game does not exactly
characterize the expressive power of the logic, but it provides a
sufficient condition for indistinguishability of structures. 

The game we define is an extension of the classic
Ehrenfeucht-Fra\"{\i}ss\'e game for first-order logic played on a pair
of structures $\Afrak$ and $\Bfrak$, where Spoiler aims to demonstrate
a difference between the two structures and Duplicator aims to
demonstrate that they are not distinguishable.  Before
introducing it formally, we make some observations.  Because our 
logic includes counting, the games are based on bijection games, which
characterise first-order logic with counting.  Thus, our games have
two kinds of moves---bijection moves, which account for counting and
ordinary quantifers, and what we call \emph{graph moves} which account
for the $\lrec$ operator.

In a graph move, Spoiler chooses an interpretation which defines a
labelled semi-graph from each of $\Afrak$ and $\Bfrak$.  An auxilliary game is then played on the
pair of semi-graphs in which Spoiler aims to show that the two
semi-graphs are distinguished by the query $\hat{\chi}$ while
Duplicator attempts to hide the difference between them.  At any point
in this auxilliary game, Spoiler can choose to revert to the main
game, for instance if a position is reached from where Spoiler can win
the ordinary bijection game.

It should be noted that in requiring Spoiler to provide an explicit
interpretation in the graph game, we build in definability in the
logic $\LREC_=$ into the rules of the game.  In that sense, it does
not provide an independent characterization of definability.
Nonetheless, it does give a sufficient criterion for proving
undefinability.  Similarly, Theorem~\ref{game-thm} below only gives
one direction of the connection between the game and the logic.  That
is, it shows that if there is a formula distinguishing $\Afrak$ and
$\Bfrak$, it yields a winning strategy for Spoiler.  We do not claim,
and do not need, the other direction.

A final remark is that we have, for simplicity, only defined the game
for the special case where $\Afrak$ and $\Bfrak$ are structures over
the same universe.  Again, this is sufficient for our purpose as the
structures we construct on which we play the game satisfy this
condition.

The game is parameterized by the two measures of size of formulas we
introduced in Section~\ref{sec:logic-def}: the rank and iteration
degree.   To
define the game, we first introduce some notation for partial maps.

For any sets $A,B, C$ and $D$ with $A\subseteq B$ and $C \subseteq D$, and function $f:A\to C$, we say $f$ is a \emph{partial function} between $B$ and $D$, which we can denote by $f:B\hookrightarrow D$. We say that a function $f': B \hookrightarrow D$ \emph{extends} $f$ if $f'|_{A} = f$. Let $\dom{f}$ denote the set $A$ where $f$ is defined. For partial function $g:B\hookrightarrow D$ with $g(a) = f(a)$ for $a \in \dom{f}\cap\dom{g}$, let $f\cup g : B\hookrightarrow D$ be the partial function defined on $\dom{f}\cup\dom{g}$ with $[f\cup g](a) = f(a)$ if $a$ is in $\dom{f}$ and $[f\cup g](a) = g(a)$ otherwise.

\subsection{Games}
We now give the formal definition of the game and prove its adequacy.

\begin{definition}\label{def:game}
The \emph{$k$-step, $q$-degree $\LREC_=$ game}  over vocabulary $\tau$
is played by two players,  \emph{Spoiler} and \emph{Duplicator} on two
$\tau$-structures $\Afrak$ and $\Bfrak$ of same universe $U$ of size
$n$.  At any stage of the game, the game position consists of a
partial injection $f:U \hookrightarrow U$ where $f(e^\Afrak) =
f(e^\Bfrak)$ for all constants $e$ in $\tau$, and the domain of $f$
has at most $k$ elements.

Let $\beta = \{b_1,b_2,\ldots,b_{|\beta|}\}$ be the domain of $f$, and
fix an ordering on it so that $\vec{b}$ is the tuple
$(b_1,b_2,\ldots,b_{|\beta|})$.  As long as $|\beta| < k$, Spoiler can
play either an extension move or a graph move.  These are defined as follows.

\begin{itemize}
    \item \emph{Extension move}: If  Spoiler selects this move then
      the duplicator begins by choosing a bijection $g:U\to U$ that
      extends $f$. The spoiler then picks some $a \in U$.  The
      resulting position is $f' = f\cup g|_{\{a\}}$.
    \item \emph{Graph move}:  In this move  Spoiler begins by
      selecting $c \leq (k-|\beta|)/2$,  and $\LREC_=$ queries
      $\varphi_E(\vec{x},\vec{y})$, and
      $\varphi_\sim(\vec{x},\vec{y})$ where $\vec{x}$ and $\vec{y}$
      are tuples of variables of length $c$. Both queries have
      iteration-degree at most $q$ and rank at most $k-|\beta|-2c$.
      Let  $V$ be $(U\cup\Ncal(U))^c$, $G_\Afrak = (V, E_\Afrak,
      \sim_\Afrak)$ be the semi-graph defined in $(\Afrak,\vec{b})$
      by the interpretation of $(\varphi_F,\varphi_\sim)$, and
      $G_\Bfrak = (V, E_\Bfrak, \sim_\Bfrak)$ be the corresponding
      semi-graph obtained from $(\Bfrak, f(\vec{b}))$.  The two
      players now play a game on these structures through a series of
      rounds.  At each round $i$, the position of this game consists
      of a tuple $\vec{a}_i \in V$, a partial injection $h_i: U
      \hookrightarrow U$ with domain of size at most $c$, and a value
      $\ell_i \leq n^q$.  Initially, Spoiler chooses $\vec{a}_0 \in
      (\beta\cup\Ncal(U))^c$ and $\ell_0 \leq n^q$ and $h_0$ is set to
      the empty map.

     If $\ell_i = 0$, then the graph move ends and the main game
      continues from the new position $f' = f \cup h_i$.  Also,
      Spoiler may choose to end the graph move and continue the main
      game from the position  $f' = f \cup h_i$.  Otherwise we
      continue to move $i+1$ which proceeds as follows
\begin{enumerate}
        \item  Duplicator begins by choosing a partial bijection
          $g_i:U\to U$ with the property that
          $g_i^{-1}(f_i(\vec{a}_i)) \cong_\Afrak \vec{a}_i$ where $f_i
          = f \cup h_i$.
        \item For each $Y \subseteq U$ with $|Y| \leq c$, Duplicator
          chooses an injection  $h_Y:Y\to X$ satisfying the following
          conditions for all $\vec{u},\vec{v} \in V$.  Here, for any $\vec{v} \in V$, we write
          $U(\vec{v})$ to denote the set of elements of $U$ that occur
          in the tuple $\vec{v}$.
          \begin{enumerate}
          \item $(g_i(\vec{a}_i), h_{U(\vec{v})}(\vec{v})) \in
            [E_\Bfrak]$ if, and only if, $(\vec{a}_i,\vec{v}) \in
            [E_\Afrak]$; 
          \item  $|[E_\Bfrak]h_{U(\vec{v})}(\vec{v})| =
            |[E_\Afrak]\vec{v}|$; and 
          \item  $h_{U(\vec{u})}(\vec{u})= h_{U(\vec{v})}(\vec{v})$
            if, and only if, $\vec{u}=\vec{v}$. 
          \end{enumerate}
          If Duplicator cannot choose such a set of partial
          bijections, then it loses the game.
        \item Spoiler chooses some $\vec{a}_{i+1} \in
          \vec{a}_i[E_\Afrak]$ and we let $\ell_{i+1} = \lfloor
          (\ell_i-1)/|[E]\vec{a_{i+1}}|\rfloor$ and  $h_{i+1} =
          h_{U(\vec{a}_{i+1})}$.
    \end{enumerate}
\end{itemize}
Spoiler wins if at any point $f$ is not a partial isomorphism between
$\Afrak$ and $\Bfrak$. Duplicator wins if $|\dom{f}|$ reaches $k$ and
Spoiler has not won.
\end{definition}

We can now show how establishing how Duplicator winning strategies in the game can be used to show inexpressibility results for $\LREC_=$.

\begin{theorem}
\label{game-thm}
Suppose $\Afrak$ and $\Bfrak$ are two $\tau$-structures with the same universe $U$,  $f:U \hookrightarrow U$ is  a partial injection with $f(e^\Afrak) = f(e^\Bfrak)$ for all constants $e$ in $\tau$ and that $\vec{a}$ enumerates the domain of $f$.  If  Duplicator has a winning strategy in the $k$-step, $q$-degree $\LREC_=$ game over $\Afrak,\Bfrak$ and $f$, then $(\Afrak,\vec{a})$ and $(\Bfrak,f(\vec{a}))$ agree on all $\LREC[\tau]$ formulas with rank at most $k-|\dom(f)|$ and iteration-degree at most $q$.
\end{theorem}
\begin{proof}
  We prove the contra-positive, so assume that  some $\LREC[\tau]$-formula
  $\varphi$ with rank at most $k-|\beta|$ and iteration-degree at most
  $q$ is true in $(\Afrak,\vec{a})$ and false in
  $(\Bfrak,f(\vec{a}))$. If $k - |\beta| = 0$ then  $f$ 
  does not induce a partial isomorphism between $\Afrak$ and
  $\Bfrak$, since otherwise the quantifier-free formula $\phi$ would
  distinguish $\vec{a}$ from $f(\vec{a})$.  Suppose then that $k- |X|
  \geq 1$. If $\varphi$ is of the form $\exists x [\psi(x)]$, $\forall
  x [\psi(x)]$ or $[\# x \psi(x)] = \mu$, then Spoiler plays the
  extension move.  Since $\varphi$ is true in $(\Afrak,\vec{a})$ and false in
  $(\Bfrak,f(\vec{a}))$, the number of elements satisfying $\psi$ in
  the two structures is different.  Hence, for any bijection $g$ extending
  $f$ that Duplicator might pick, there is some  some $u \in U$ such
  that $(\Afrak,\vec{a},u)$ and $(\Bfrak,f(\vec{a}),g(y))$ do not
  agree on $\psi(x)$.  That is to say, either $\psi$ or $\neg \psi$ is
  a formula true in the former but false in the latter.  Either of
  these is a formula of rank at most $k-|\beta| - 1$ and so if 
  Spoiler chooses $u$ and extends $f$ to
  $f\cup g|_{\{u\}}$ it wins by the induction hypothesis.

Assume next that $\varphi$ is of the form $$[\lrec_{\vec{x},\vec{y},\vec{p}} \varphi_E,\varphi_\sim, \varphi_C](\vec{w},\vec{r})$$
where the lengths of $\vec{x}$ and $\vec{y}$ are both $c \leq
(k-|\beta|)/2$ andthe length of $\vec{r}$ is $q' \leq q$. To simplify
notation, let $\vec{a}_\Afrak$ denote $\vec{a}$ and $\vec{a}_\Bfrak$
denote $f(\vec{a})$. Let $V$ be $(U\cup\Ncal(U))^c$,  and for  $\Sfrak
\in \{\Afrak,\Bfrak\}$, let $G_\Sfrak = (V,
E_\Sfrak, \sim_\Sfrak)$ be the semi-graph obtained by the
interpretation $(\varphi_F,\varphi_\sim)$ in 
$(\Sfrak,\vec{a}_\Sfrak)$ and $C_\Sfrak$ the labelling defined on it
by  $\varphi_C$.  For any $\vec{a} \in V$, we write $\tilde{a}$ for the set of elements that occur in $\vec{a}$ and $U(\tilde{a})$ for the set of those elements of $\tilde{u}$ that are in $U$.

 Spoiler then chooses the starting node $\vec{a}_0$ to be the
 $(\Afrak,\vec{a})$ interpretation of $\vec{w}$, and $\ell_0$ to be
 the value $\langle \vec{s} \rangle$ where $\vec{s}$ is the
 interpretation in $(\Afrak,\vec{a})$ of $\vec{r}$.  Note that by our
 choice of $\varphi$, it is true that $(\vec{a}_0,\ell_0) \in
 \chi[G_\Afrak,\langle C_\Afrak \rangle]$ and  $(f(\vec{a}_0),\ell_0)
 \not\in \chi[G_\Bfrak,\langle C_\Bfrak \rangle]$.  We prove by
 induction on $i$ that if $(\vec{a}_i,\ell_i) \in
 \chi[G_\Afrak,\langle C_\Afrak \rangle]$ if, and only if,
 $(f(\vec{a}_i),\ell_i) \not\in \chi[G_\Bfrak,\langle C_\Bfrak
 \rangle]$ then  Spoiler has a winning strategy in the game move.
 Assume without loss of generality that $(\vec{a}_i,\ell_i)$ is in
 $\chi[G_\Afrak,\langle C_\Afrak \rangle]$.  For $\Sfrak \in \{\Afrak,\Bfrak\}$, let $m_\Sfrak$ be the number
$$|\{v \in u[F_\Sfrak] \mid (v, \lfloor (\ell_i -
1)/|[F_\Sfrak]v|\rfloor) \in \chi(G_\Sfrak,C_\Sfrak)\}|.$$
Then,either $m_\Afrak \neq m_\Bfrak$ or $m_\Afrak \in C_\Afrak$ but
$m_\Bfrak \notin C_\Bfrak$. So, if $m_\Afrak = m_\Bfrak$  we must have
the latter case, where  Spoiler can simply reset $f$ to $f_i$ and
proceed by playing the game on the structures
$(\Afrak,\vec{a},\vec{a}_i,\vec{s})$ and
$(\Bfrak,f(\vec{a}),f(\vec{a}_i),\vec{s})$ where $\langle \vec{s}
\rangle = m_\Afrak = m_\Bfrak$. Here we also have the base case --- if
$\ell_i = 0$ then $m_\Afrak = m_\Bfrak = 0$. If, instead, $m_\Afrak
\neq m_\Bfrak$ then  Spoiler proceeds to the next step of the
iteration. If  Duplicator is able to come up with a valid set of
partial bijections in step 3 of the iteration, there must be some
$\vec{b} \in \vec{a}_i[F]$ with $(\vec{b},\ell') \in
\chi(G_\Afrak,C_\Afrak)$ if and only if
$(h_{U(\tilde{b})}(\vec{b}),\ell') \notin \chi(G_\Bfrak,C_\Bfrak)$,
where $\ell' = \lfloor(\ell-1)/|[F_\Afrak]\vec{b}|\rfloor$ (which is
also equal to
$\lfloor(\ell-1)/|[F_\Bfrak]h_{U(\tilde{b})}(\vec{b})|\rfloor$
according to the conditions Duplicator's choice must satisfy). Thus, by the induction hypothesis,  Spoiler has a winning strategy if it picks $\vec{a}_i$ to be $\vec{b}$.
\end{proof}

\subsection{Structures}
Here we describe the particular instances of the path systems problem
for which we construct Duplicator winning strategies in the game we
have just defined, and outline the strategy.

Consider a tree $T = (V,E)$ to be defined as a directed graph
where edges are oriented from a parent to its children.
The structures we consider are obtained by taking the product of a
complete binary tree $T$ with a large cyclic group of prime order
$\gf{p}$. 
The unary relation $S$ then encodes an assignment of values
in $\gf{p}$ to the leaves of the tree and the ternary relation $R$ is
chosen so that determining whether the target $\tgt$ is in the upward
closure of $S$ amounts to summing these values along the tree.
We give a formal definition for future reference.
\begin{definition}\label{def:structures}
  For a positive integer $h$, let $T = (V,E)$ be the complete binary tree
  of height $h$, $L \subseteq V$ the set of its leaves and $\textbf{root} \in V$
  its root.  For any
  prime $p$, function $s: L \to \gf{p}$ and element $t \in \gf{p}$
  define the structure $\Pcal(h,p,s,t)$ to be the instance of $\PSP$ with
  \begin{enumerate}
  \item universe $V \times \gf{p}$
  \item $ \left((u,a),(v,b),(w,c)\right) \in R \iff
    a + b = c \text{ AND } E(w,u) \text{ AND } E(w,v)$
  \item $S = \{(l,\sigma(l) \mid l \in L \}$
  \item $\tgt = (\textbf{root},t).$
  \end{enumerate}
\end{definition}
Note that for each node $v \in V$ there is a unique value of $a \in
\gf{p}$ for which $(v,a)$ is in the upward closure of $S$.  To be
precise, $a$ is the sum (modulo $p$) of all the values of $\sigma(l)$
for leaves $l$ below $v$ in the tree $T$.  In particular
$\Pcal(h,p,s,t)$ is a positive instance of $\PSP$ if, and only if, $t
= \sum_{l \in L} \sigma(l)$.  What we aim to show is that two
structures $\Pcal(h,p,s,t)$ and  $\Pcal(h,p,s,t')$ with $t \neq t'$ 
are
indistinguishable by $\LREC_=$ formulas of rank $k$ and degree $q$ as
long as $h$ and  $p$ are large enough with respect to $k$ and $q$.

Intuitively it is clear that a fixed-point definition of the upwards
closure of $S$, such as given by the $\LFP$ formula will determine in
$h$ (i.e.\ the height of the tree) iterations whether $\tgt$ is in the
closure of $S$.  What makes this difficult for $\LREC$ is that at each
element $u \in U$ there are $p$ distinct pairs of elements $v,w \in U$
for which $(u,v,w) \in R$.  The number of paths multiplies giving
$p^h$ distinct ways of reaching $\tgt$, and by an appropriate choice
of $p$ and $h$, we can ensure that this is not bounded by a polynomial
in the number of elements of $U$, which is $p2^h$.  Moreover, we
cannot eliminate the multiplicity of paths by taking a suitable quotient.
Of course, this
only shows that the obvious inductive method of computing the upward
closure of $S$ cannot be implemented in $\LREC_=$.  To give a full
proof, we have to consider all other ways that this might be defined,
and that is the role of the game.

In the winning strategy we describe, Duplicator plays particular
bijections which we now describe.  Let $\Afrak$ and $\Bfrak$ be two
instances of $\PSP$ obtained as described above from a  tree $T = (V,
E)$ and prime $p$, with the same set $S$ and different values of
$\tgt$.  Note, in particular, that $\Afrak$ and $\Bfrak$ have the same
universe $U = V \times \gf{p}$.

For any set $X \subseteq V$.  A function $\rho: X \to \gf{p}$ induces
a bijection $\bij(\rho) : X \times \gf{p} \to X \times \gf{p}$ given
by $\bij(\rho)(v,a) = (v,a+\rho(v))$.  We extend this to a bijection
on the set $\Dcal(X) = $ given by  $f(v,a) := (v,a+\rho(a))$ for all $v
\in X$ and $(X\times \gf{p})\cup \Ncal(U)$, by letting it be the
identity on all elements of $\Ncal(U)$.  In the Duplicator strategies
we describe, all bijections played are of this form.  Thus, we usually
describe them just by specifying the function $\rho$ which we call the
\emph{offset function}.  We abuse terminology somewhat and say that
$\rho$ is a partial isomorphism from $\Afrak$ to $\Bfrak$ to mean that
$\bij(\rho)$ induces a partial isomorphism.

Roughly speaking, Duplicator's winning strategy is to play offset
functions which are zero at the leaves of the tree $T$ and offset by
the difference between the values of $\tgt$ in $\Afrak$ and $\Bfrak$
at the root.  Spoiler has to try and expose this inconsistency by
building a path between the root and the leaves.  We show that
Duplicator can maintain a height (depending on the parameters $k$ and
$q$) below which it plays offsets of zero, without the inconsistency
being exposed.  Of course, Duplicator has to respond to graph moves in
the game, so we have to consider paths in the interpreted semigraphs,
where each node may involve elements from many different heights in
the tree.  This is what makes describing the Duplicator winning
strategy challenging.  In the next section we develop the tools for
describing it.

%% file: extender.tex
The Duplicator strategies we describe in the next section for the $\LREC$ game rely on certain combinatorial properties of the complete binary tree $T$ and certain offset functions $\rho: T \ra \gf{p}$.  In this section we develop some combinatorial properties of such trees and functions that allow us to effectively describe the strategies.

Let $T=(V,E)$ be a complete (directed) binary tree with $N = 2^n$ leaves.  The \emph{height} of a node $v\in V$, denoted $\hgt{v}$ is the distance of $v$ to a leaf (since $T$ is complete, this is the same for all leaves reachable from $v$).  Thus, if $v$ is a leaf its height is $0$ and if $v$ is the root of $T$, its height is $n$.

Let $R_T$ be the set of triples $(x,y,z)$ such that $E(x,y)$ and $E(x,z)$ and let $R  = \{ \{x,y,z\} \mid (x,y,z) \in R_T\}$.  In other words $R$ is the collection of unordered sets of three elements consisting of a node of $T$ along with its two children.  We say that a three element set $\{x,y,z\}$ is \emph{related} if $\{x,y,z\} \in R$.  We call an element of $R$ a \emph{related triple}.

Say a set $X \subseteq V$ is \emph{closed} if for every related triple $\{x,y,z\} \in R$, if $x,y \in X$ then $z \in X$.  For every $X\subseteq V$, there is a unique minimal closed set $\cl{X}$ such that $X \subseteq \cl{X}$.  We call $\cl{X}$ the \emph{closure} of $X$.

Let $\mnh{X}$ denote the minimum height of any element of $X$ and $\mxh{X}$ denote the maximum height of any element of $X$.

\begin{proposition}
  $\mnh{\cl{X}} = \mnh{X}$
\end{proposition}
\begin{proof}
  Consider the sequence of sets given by $X_0 = X$ and $X_{i+1} = X_i \cup \{ z \mid \{x,y,z\} \in R \text{ for some } x,y \in X_i \}$.  Then $\cl{X} = \bigcup_i X_i$.  A simple induction on $i$ shows that $\mnh{X_{i+1}} = \mnh{X_i}$ for all $i$, establishing the claim.
%
\end{proof}

For any pair of vertices $x,y \in V$, there is a unique undirected path $z_1,\ldots,z_k \in V$  such that $x =z_1$, $y = z_k$ and for all $i$ with $1 \leq i < k$, we have $E(z_i,z_{i+1})$ or $E(z_{i+1},z_{i})$.  Say that a set $X \subseteq V$ is \emph{connected} 
whenever $x,y \in X$ the undirected path from $x$ to $y$ is contained in $X$.

A \emph{connected component} of a closed set $X$ is a maximal closed connected subset of $X$.   It is clear that every closed set $X$ is a disjoint union of closed connected components.

For $v \in V$, we say a set $X \subseteq V$ \emph{encloses} $v$ if every path from $v$ to a leaf goes through an element of $X$.  We say that $X$ \emph{minimally encloses} $v$ if $X$ encloses $v$ and no proper subset of $X$ encloses $v$.

If $X$ is a closed connected set, there is a unique element $r \in X$ with  $\hgt{r} = \mxh{X}$.  Moreover, there is a set $F \subseteq X$ which minimally encloses $r$ and such that $X$ is exactly the set of elements $y \in V$ such that $y$ is on the path from $r$ to $f$ for some $f \in F$.  We call $F$ the \emph{frontier} of $X$.  Note that this allows the trivial case when $X = F = \{r\}$.  In all cases, we have $\overline{F} = X$.  We write $\hgt{X}$ for $\hgt{r} = \mnh{X}$.

\begin{proposition}
  If $X$ is a closed connected set with root $r$ and frontier $F$ then
  $$ |F| > \hgt{X}.$$
\end{proposition}
\begin{proof}
  Let $h = \hgt{X}$.  The proof proceeds by induction on $h$.  If $h = 0$, then $F = \{r\}$ and the inequality is satisfied.  Suppose then that $h > 0$, and so there is some element in $X$ whose height is less than that of $r$.  Since $X$ is connected this implies that some child of $r$ is in $X$ and since $X$ is closed this implies both children of $r$ are in $X$.  Let the two children of $r$ be $s_1$ and $s_2$.  Then, we have $X = X_1 \cup X_2 \cup\{r\}$ where $X_1$ and $X_2$ are closed connected sets with roots $s_1$ and $s_2$ respectively.  Let $F_1$ and $F_2$ be the respective frontiers of $X_1$ and $X_2$ and note both of these are non-empty.  Since $\mnh{X_1} \geq \mnh{X}$ and $\hgt{s_1} = \hgt{r} -1$, we have $\hgt{X_1} < h$ and so by induction hypothesis $|F_1| > \hgt{X_1}$.  Since $F_2$ is non-empty, we have
  $$|F| = |F_1| + |F_2| >  \hgt{s_1} + 1 - \mnh{X_1} \geq \hgt{r} - \mnh{X} = \hgt{X}$$
  as required.
\end{proof}

For a closed set $X \subseteq V$, we say a function $\rho: X \ra \Zp$ is \emph{consistent} if whenever $x,y,z  \in X$ are such that $E(x,y)$ and $E(x,z)$, then $\rho(x) = \rho(y) + \rho(z)$.   If $X \subseteq Y$, we say that a function $\rho': Y \ra \Zp$ \emph{extends} $\rho$ if $\rho'|_{X} = \rho$.  The following is a useful characterization of consistent functions on closed sets.
\begin{proposition}\label{prop:consistent}
  For a closed set $X$,  $\rho: X \ra \Zp$ is consistent if, and only if, for every $F \subseteq X$ and $x\in X$ such that $F$ minimally encloses $x$, we have $\rho(x) = \sum_{y \in F} \rho(y)$.
\end{proposition}
\begin{proof}
  The direction from right to left is immediate, since for any $x,y,z \in X$ with $(x,y,z) \in R_T$, we have that $F = \{y,z\}$ minimally encloses $x$.  Thus, by assumption $\rho(x) = \sum_{w \in F} \rho(w) = \rho(y) + \rho(z),$ as required.

  In the other direction, assume that $\rho$ is consistent and suppose $F$ minimally encloses $x$.  Let $h = \hgt{\cl{F}}$ and we proceed by induction on $h$.  If $h=0$ then $F = \{x\}$ and  $\rho(x) = \sum_{y \in F} \rho(y)$ is clearly true.  Suppose then the claim is true for all $F'$ with $\hgt{\cl{F'}} \leq h$ and let $F$ be a set minimally enclosing $x$ with $\hgt{F} = h+1$.  Let $s_1$ and $s_2$ be the two children of $x$ and $F_1$ and $F_2$ the subsets of $F$ that minimally enclose $s_1$ and $s_2$ respectively.  Then $ \sum_{y \in F} \rho(y) = \left( \sum_{y \in F_1} \rho(y) \right) + \left( \sum_{y \in F_2} \rho(y) \right) = \rho(s_1) + \rho(s_2) = \rho(x)$.  Here the first equality holds from the fact that $F_1$ and $F_2$ form a partition of $F$, the second by induction hypothesis and the third by the consistency of $\rho$.
\end{proof}

For any (not necessarily closed)  set $X \subseteq V$, we say a function $\rho: X \ra \Zp$ is \emph{consistent} if there is a consistent $\rho': \cl{X} \ra \Zp$ which extends $\rho$.  Note that if such a $\rho'$ exists then it is unique.   We can extend Proposition~\ref{prop:consistent} to sets $X$ that are not closed.
\begin{proposition}\label{prop:non-closed}
  For any set $X$, $\rho: X \ra \Zp$ is consistent if, and only if, for every $F \subseteq X$ and $x\in X$ such that $F$ minimally encloses $x$, we have $\rho(x) = \sum_{y \in F} \rho(y)$.
\end{proposition}
\begin{proof}
  In one direction, suppose $\rho$ is consistent and there is an $F \subseteq X$ minimally enclosing $x \in X$.    Then this $F$ and $x$ is present in $\cl{X}$ and so, by the consistency of  $\rho$ we have $\rho(x) = \sum_{y \in F} \rho(y)$ by Proposition~\ref{prop:consistent}.

  In the other direction,  suppose $\rho(x) = \sum_{y \in F} \rho(y)$ for all  $F$ minimally enclosing $x$ in $X$.  If $X$ is closed, there is nothing to prove by Proposition~\ref{prop:consistent}.  If not, then we can enumerate the elements of $\cl{X}$ that are not in $X$ in an order $x_1,\ldots,x_c$ so that for each $x_{i+1}$, there are elements $y,z \in X \cup \{x_1,\ldots,x_i\}$ such that $\{x_{i+1},y,z\} \in R$.  Fix such an order and write $X_i$ for the set $X \cup \{x_1,\ldots,x_i\}$.  We prove by induction on $i$ that we can extend $\rho$ to  a function $\rho_i$ on $X_i$ such that $\rho_i(x) = \sum_{y \in F} \rho_i(y)$ for all  $F$ minimally enclosing $x$ in $X_i$.  Since $X_n = \cl{X}$ is closed, the result then follows by Proposition~\ref{prop:consistent}.

  We let $\rho_0= \rho$ and the base case follows by assumption.  Suppose $\rho_i$ has been defined to satisfy the inductive hypothesis.  If there is exactly one pair of elements $y,z \in X_i$ such that $\{x_{i+1},y,z\} \in R$ then we extend $\rho_i$ to $x_{i+1}$ in the natural way.  That is if $y$ and $z$ are the two children of $x_{i+1}$ we let $\rho_{i+1}(x_{i+1}) = \rho_i(y) + \rho_i(z)$ and if $y$ is the parent and $z$ the sibling of $x_{i+1}$ we let $\rho_{i+1}(x_{i+1}) = \rho_i(y) - \rho_i(z)$.  In each case, we can verify that $\rho_{i+1}$ satisfies the condition that $\rho_{i+1}(x) = \sum_{w \in F} \rho_{i+1}(w)$ for all  $F$ minimally enclosing $x$ in $X_{i+1}$.  If $F \cup\{x\}$ does not contain $x_{i+1}$, there is nothing to check as the condition is satisfied by induction hypothesis.  Similarly, if $F = \{x\}$ the condition is trivially satisfied.  Otherwise, there are a number of cases to be considered.
  \begin{itemize}
  \item  If $x = x_{i+1}$ and $y$ and $z$ are the two children of $x_{i+1}$ then $F$ can be partitioned into two sets $F_1$ and $F_2$ minimally enclosing $y$ and $z$ respectively.  By induction hypothesis we have $\rho_i(y) = \sum_{w \in F_1} \rho_i(w)$ and  $\rho_i(z) = \sum_{w \in F_2} \rho_i(w)$ and so $\rho_{i+1}(x) = \sum_{w \in F} \rho_{i+1}(w)$ as desired. 
  \item If $x = x_{i+1}$ and $y$ is the parent and $z$ the sibling of $x_{i+1}$, then $F \cup \{z\}$ minimally encloses $y$ and so by induction hypothesis  $\rho_i(y) = \rho_i(z) + \sum_{w \in F} \rho_i(w)$.  Thus, setting $\rho_{i+1}(x_{i+1}) = \rho_i(y) - \rho_i(z)$ gives the desired result.
  \item If $x_{i+1} \in F$ and $y$ and $z$ are the two children of $x_{i+1}$ we have that $F\setminus \{x_{i+1}\} \cup \{y,z\}$ minimally encloses $x$ and we have $\rho_{i+1}(x) = \rho_i(x) = \rho_i(y) + \rho_i(z) + \sum_{w \in F \setminus \{x_{i+1}\}} \rho_i(w) = \rho_{i+1}(x_{i+1}) + \sum_{w \in F \setminus \{x_{i+1}\}} \rho_i(w) = \sum_{w \in F} \rho_{i+1}(w).$
  \item  If $x_{i+1} \in F$ and $y$ is the parent and $z$ the sibling of $x_{i+1}$, then it must be that $z \in F$ and also that $F \setminus \{z,x_{i+1}\} \cup \{y\}$ minimally encloses $x$.  From these two facts it is easily seen that $\rho_{i+1}(x) = \rho_i(x) = \sum_{w \in F \setminus \{z,x_{i+1}\} \cup \{y\}}\rho_i(w) = \rho_i(y) - \rho_I(z)+ \sum_{w \in F \setminus \{x_{i+1}\} }\rho_i(w)  = \rho_{i+1}(x_{i+1}) + \sum_{w \in F \setminus \{x_{i+1}\} }\rho_i(w)  = \sum_{w \in F}\rho_{i+1}(w).$
  \end{itemize}

  The only case remaining is if both children of $x_{i+1}$, say $y_1$ and $z_1$ as well as its parent, say $y_2$ and its sibling, say $z_2$ are all in $X_i$.   Then we can consistently define $\rho_{i+1}$ as before as long as $\rho_i(y_1) + \rho_i(z_1) = \rho_i(y_2) - \rho_i(z_2)$.  But, since $\{y_1,z_1,z_2\}$ minimally encloses $y_2$, this is guaranteed by the induction hypothesis.
 \end{proof}

The situation we are often interested in is when we have a consistent function $\rho: X \ra \Zp$ and we wish to extend it to a consistent function $\sigma: Y \ra \Zp$ for some $Y \supseteq X$ in a \emph{minimal} fashion.  That is to say, we want to set  $\sigma$ to be $0$ wherever possible.  To this end, we make the following definition.

\begin{definition}\label{def:free}
  Suppose $X,Y \subseteq V$ with $X \subseteq Y$ and $\rho: X \ra \Zp$ is a consistent function.  We say that $y \in \cl{Y}$ is \emph{free} over $\rho$ if for every $S \subseteq \cl{Y}$ that is closed and connected with head $x$ and frontier $F$ and $y \in F \cup \{x\}$: either \emph{(i)}  $\rho(z) = 0$ for each $z \in X \cap ( F \cup \{x\})$  ; or \emph{(ii)} there is a $w \in S \cap X$ such that $w \not\in F \cup \{x\}$.
\end{definition}

We now show that a consistent $\rho$ on the set $X$ can be extended to
a consistent offset function on $Y  \supseteq X$ which is zero on all
elements that are free over $\rho$.

\begin{proposition}\label{prop:forced}
  Let $X,Y \subseteq V$ with $X \subseteq Y$ and $\rho: X \ra \Zp$ be a consistent function.  Suppose $Z \subseteq \cl{Y}$ is the set of elements in $\cl{Y}$ which are free over $\rho$ and $\sigma: Z \ra \Zp$ the function that takes all such elements to $0$.  Then $\rho \cup \sigma$ is consistent.
\end{proposition}
\begin{proof}
  Assume, for contradiction, that $\rho' := \rho \cup \sigma$ is not
  consistent.  Then, by Proposition~\ref{prop:non-closed}, there is a
  closed connected set $S \subseteq X \cup Z$ with frontier $F$ and head $x$ such that $\rho'(x) \neq \sum_{y \in F} \rho'(y)$.   Let $S$ be a minimal such set.  Note that $F \cup \{x\}$ must contain an element of $Z$, since $\rho$ is assumed to be consistent.  Then, it must be the case that $\rho(z) \neq 0$ for some $z \in X \cap ( F \cup \{x\})$ for otherwise by definition we would have $\rho'(z) = 0$ for all $z \in F \cup\{x\}$ and so $\rho'(x) = \sum_{y \in F} \rho'(y)$.
Then, by Definition~\ref{def:free}, there is a  $w \in S \cap X$ such that $w \not\in F \cup \{x\}$.  Let $F' \subseteq F$ be the unique subset of $F$ that minimally encloses $w$.  We then have two cases: either $\rho'(w) = \rho(w) \neq \sum_{y \in F'} \rho'(y)$ or $\rho'(w) = \rho(w) = \sum_{y \in F'} \rho'(y)$.  In the first case, let $S' = \cl{F' \cup\{w\}}$ and in the second, let $S' = \cl{F\setminus F' \cup \{w,x\}}$.  In either case we have $S'$ is a proper closed connected subset of $S$ with the property that $\rho'$ at its head is not the sum of the values of $\rho'$ at its fronitier, contradicting the minimality of $S$.
\end{proof}

We now establish the main combinatorial property of consistent offset
functions which will enable us to construct Duplicator winning
strategies for the graph move of the game.

\begin{proposition}\label{prop:sequence}
  Say we are given a set $X \subseteq V$ along with a consistent $\rho: X \ra \Zp$, and let $Y_1,\ldots,Y_r \subseteq V$ be sets with $|Y_i| \leq s$ for all $i$.  Then, there is a sequence of functions $\sigma_i: Y_i \rightarrow \Zp$ such that:
  \begin{enumerate}
  \item for each $i$ with $1\leq i < r$, $\rho \cup \sigma_i \cup \sigma_{i+1}$ is consistent; and 
  \item for any $i$ and  $y \in Y_i$, if $\sigma_i(y) \neq 0$ then there is some $x\in X$ with $\rho(x) \neq 0$ and $\hgt{y} \geq  \hgt{x} - 2(|X| + s)$.
  \end{enumerate}
\end{proposition}
\begin{proof}
  For any $v \in V$, we write $T(v)$ for the subtree of $T$ rooted at $v$.  We now define, for each $i$ a set $H_i \subseteq V$, which we use in defining the function $\sigma_i$.
The set $H_i$ for $1 \leq i \leq r$ is the set of all vertices $u \in V$ such that all of the following conditions hold:
\begin{itemize}
\item $u$ has a grandparent $v$ and no element $x$ of $X$ with non-zero $\rho(x)$ appears in $T(v)$;
\item for some $j \geq i$, all elements of $Y_j\cup X$ that are in $T(v)$ are in $T(u)$; and 
\item for all $j'$ with $i \leq j' < j$, there is no grandchild $w$ of $v$ such that all elements of $Y_{j'}\cup X$ that are in $T(v)$ are in $T(w)$.
\end{itemize}
Note that by construction $H_i$ cannot contain a pair of sibling nodes.

Let $\rho_0 = \rho$ and for each $i  \geq 1$ define $\rho_i: X \cup H_i \ra \Zp$ to be the function such that $\rho_i(x) = \rho(x)$ for $x \in X$ and $\rho_i(h) = 0$ for all $h\in H_i$.  It easily follows from Proposition~\ref{prop:forced} that $\rho_i$ is consistent for all $i$.

We now define the series of functions $\sigma_i$ by induction on $i$ having the following two properties
\begin{enumerate}
\item $\rho_i \cup \sigma_i$ is consistent; and 
\item whenever $y \in T(h)$ for some $h \in H_i$ and $y\in Y_i$, then $\sigma_i(y) = 0$.
\end{enumerate}

For a base case, define $Y_0$ to be the empty set and $\sigma_0$ to be the empty function.  Clearly the above two properties are trivially satisfied.  Assume, by induction that $\sigma_i$ has been defined with the above properties.  We first show that $\eta = \rho_{i+1} \cup \sigma_i$ is consistent.   Suppose, towards a contradiction that it is not.  Then, by Proposition~\ref{prop:consistent} there is an $x \in \dom(\eta) = X \cup Y_i \cup H_{i+1}$ and an $F \subseteq X \cup Y_i \cup H_{i+1}$ that minimally encloses $x$ such that $\eta(x) \neq \sum_{y \in F} \eta(y)$.  If $F \cup \{x\}$ contains no element of $H_{i+1}$, then $\eta(x) = \sum_{y \in F} \eta(y)$ by the consistency of $\rho_i \cup \sigma_i$, so suppose $F \cup \{x\}$ contains an element $h$ of $H_{i+1}$.  We argue that then $h$ is also an element of $H_i$ and therefore again $\eta(x) = \sum_{y \in F} \eta(y)$ by the consistency of $\rho_i \cup \sigma_i$.  Note that any element $h$ of $H_{i+1}$ is an element of $H_i$ unless the grandparent $u$ of $h$ has another grandchild $w \neq h$ such that all elements of $(Y_i\cup X) \cap T(u)$ are in $T(w)$.
Indeed, by definition of $H_{i+1}$, there is a $j \geq i+1$ such that all elements of $Y_j\cup X$ that are in $T(u)$  are in $T(h)$.   This $j$ also witnesses that $h$ is in $H_i$ unless $u$ has another grandchild $w \neq h$ such that all elements of $(Y_i\cup X) \cap T(u)$ are in $T(w)$.  We show by cases that this cannot happen.

\noindent
\emph{Case 1}  $h = x$.  Then, by definition of $H_{i+1}$ there are no elements of $X$ in $T(x)$ that $\rho$ takes to a non-zero element and so in particular none in $F$.  Further, since $H_{i+1}$ cannot contain a pair of sibling nodes and any frontier minimally enclosing $x$ must contain such a pair, we conclude that $F$ contains at least one element of $Y_i\cup X$.    Hence, the grandparent $u$ of $h$ cannot have a grandchild $w$ distinct from $h$ such that $T(w)$ contains $T(u) \cap (Y_i\cup X)$.

\noindent
\emph{Case 2}: $x \not\in H_{i+1}$ is the parent of $h$.   Let $u$ be the parent of $x$.  Since $T(u)$ does not contain any element of $X$ that $\rho$ takes to a non-zero element and $x \not \in H_{i+1}$ we conclude that $x \in Y_i\cup X$.  Thus, there is no grandchild $w$ of $u$ such that $T(w)$ contains $T(u) \cap (Y_i\cup X)$.

\noindent
\emph{Case 3:}  $x \not\in H_{i+1}$ is not the parent, but an ancestor of $h$.  Then, consider $y$, the sibling of $h$ and $z$ the sibling of the parent of $h$.  We know $y \not\in H_{i+1}$ since two sibling nodes cannot be in $H_{i+1}$.  Since $F \cap T(y)$ must enclose $y$, we conclude that at least one element of $F \cap T(y)$ is not in $H_{i+1}$ and so is in $Y_i\cup X$.  By the same argument, there is an element of $Y_i\cup X$ in $T(z)$.  Thus, the grandparent $u$ of $h$ does not have a unique grandchild $w$ such that $T(w)$ contains $T(u) \cap (Y_i\cup X)$.

Thus, $\eta$ is consistent.   Let $Z$ be the set of elements in $C = \cl{X \cup Y_i \cup Y_{i+1} \cup H_{i+1}}$ that are free over $\eta$.  Then by Proposition~\ref{prop:forced} there is a consistent extension $\eta': C \ra \Zp$ of $\eta$ such that $\eta'(x) = 0$ for all $x \in Z$.  We define $\sigma_{i+1}$ to be the restriction of $\eta'$ to $Y_{i+1}$.

We first argue that $\sigma_{i+1}$ satisfies the two properties required inductively.  That $\rho_{i+1} \cup \sigma_{i+1}$ is consistent follows from the consistency of $\eta'$.  For the second, it suffices to show that whenever $y \in T(h) \cap Y_{i+1}$ for some $h \in H_{i+1}$, then $y$ is free over $\eta$.
For any element $ z \in \dom(\eta) \cap T(h)$ we have either $z \in Y_i\cup X$ or $z \in H_{i+1}$.  Note that $T(h)$ contains no element of $X$ that gets taken to a non-zero element of $\rho$ by definition of $H_{i+1}$ and so in the latter case, $\eta'(z) = 0$ again by definition.  In the former case, we have an element of $X\cup Y_i$ in $T(h)$ and so $h$ must be in $H_i$, hence $\eta(z) = 0$ by induction hypothesis.  So, suppose $y \in T(h) \cap Y_{i+1}$ and let $S$ be a closed connected subset of $C$ with frontier $F$ and head $x$ such that $y \in F \cup \{x\}$.  If $x$ is not in $T(h)$, then $h \in \dom(\eta) \cap S$ and $h \not\in F \cup \{x\}$.  On the other hand, if $x \in T(h)$, then $S \subseteq T(h)$ and we have established that $\eta(z) = 0$ for all $z \in T(h) \cap \dom(\eta)$.  This shows that $y$ is free over $\eta$.

Having defined the sequence of functions $\sigma_i$, we argue that
they satisfy the two properties required in the proposition.  The
first is immediate from the definition.  Indeed, $\rho \cup \sigma_i
\cup \sigma_{i+1}$ is a restriction to $X \cup Y_i \cup Y_{i+1}$ of
the consistent function $\eta'$ defined at stage $i+1$ of the
construction above.

For the second property, suppose towards a contradiction that there is
a $y \in Y_i$ with $\sigma_i(y) \neq 0$ and such that 
the shortest
path from $y$ to an ancestor $z$ with $T(z)$ containing some $x \in X$ where $\rho(x)$ is not zero, is of
length greater than $2(s + |X|)$.  Let $y=u_0,u_1,\ldots,u_m$ be this path.
Since, $|Y_i| \leq s$, there must be a $j < m-2$ such that
$T(u_j)$, $T(u_{j+1})$ and $T(u_{j+2})$ all contain the same elements
of $Y_i \cup X$.  Since $T(u_{j+2})$ contains no element $w$ of $X$ where $\rho(w)$ is not zero, we conclude
that $u_j$ is in $H_i$.  This implies that $\sigma_i(u) = 0$ for all
$u \in T(u_j)$ and so in particular for $y$, giving us the contradiction.

\end{proof}

%% file: main-theorem.tex
In this section we prove the main result of the paper. We make use of
the tree properties established in Section~\ref{sec:extender} to show
further combinatorial results about the specific subclass of
semi-graphs that arise as interpretations from the structures defined
in Definition~\ref{def:structures}.
From this we are able to arrive at a winning strategy for Duplicator.

Let $T = (V,E)$ be a complete binary tree with $N = 2^n$ leaves, $p$
be a prime number, and fix a subset $\beta \subseteq V$.  Intuitively, these correspond to the parameters of the interpretation and remain fixed in our game argument.
For $m>0$ and any set $X$ of size $m$ let
$\Ncal(X)$ be the set $[m]$. For $X\subseteq V$, let $\Dcal(X)$ denote
the set $(X\times\gf{p})\cup \Ncal(V\times\gf{p})$. Fix a semi-graph
$G = (W, F, \sim)$ with $W \subseteq \Dcal(V)^c$ for some constant $c$
and $F$ the edge relation of $G$. Recall that $\cong$ denotes the
symmetric reflexive transitive closure of $\sim$, and $[G]$ denotes
the graph obtained from $G$ by taking the quotient with respect to
this relation.  For any $\vec{a} \in W$, we write $\tilde{a}$ for the set of elements that occur in $\vec{a}$ and $V(\tilde{a})$ for its projection on $V$, i.e.\ the set $\{v \in V \mid (v,d) \in \tilde{a} \text{ for some } d \}$.  For any set closed set $X \subseteq V$, we use the term \emph{frontier} of $X$ to mean the union of the frontiers of all the connected components of $X$.

 For $X\subseteq V$, and consistent $\rho:X\to \gf{p}$, let $\Zcal(\rho)$ be the zero-locus of $\rho$, and $\nh{\rho}$ be equal to $\mnh{X\setminus\Zcal(\rho)}-1$. For a sequence of sets $Y_1,Y_2,\ldots, Y_t\subseteq V$, let $\lift{\rho}{Y_1, Y_2, \ldots, Y_t}$ be the set of all sequences of functions $\sigma_i: Y_i\to\gf{p}$ which extend $\rho$ and satisfy the conditions of Proposition~\ref{prop:sequence}.
 Recall that $\rho$ induces a bijection $\bij(\rho) = f:\Dcal(X)\to \Dcal(X)$ which is the function where $f(v,a) := (v,a+\rho(v))$, when $(v,a) \in X\times \gf{p}$, and $f$ is the identity otherwise.  For either structure $\Afrak \in \{G,[G]\}$, we say $\rho$ \emph{induces a partial isomorphism} over $\Afrak$ if $f$ is a partial isomorphism from $\Afrak$ to itself. Moreover, we say $\rho$ \emph{induces a liftable isomorphism} over $\Afrak$ if for all sequences $Y_1,Y_2,\ldots Y_t$ with $|Y_i| \leq c$, there is some $(\sigma_1,\sigma_2,\ldots,\sigma_t) \in \lift{\rho}{Y_1, Y_2, \ldots, Y_t}$ such that $\sigma_i\cup\sigma_{i+1}$ induces a partial isomorphism over $\Afrak$ for $i<t$. For all $d$ let $f^d$ denote $f$ composed with itself $d$ times.
 
\begin{proposition}
\label{degree}
For $X \subseteq V$, function $\rho:X\to \gf{p}$ inducing a partial
isomorphism $f = \bij(\rho)$ over $[G]$ and tuple $\vec{a}\in
\Dcal(X)^c \subseteq W$, if there are $d_1,d_2 \in \gf{p}$ with
$f^{d_1}(\vec{a}) \cong f^{d_2}(\vec{a})$ then $f(\vec{a})\cong \vec{a}$.
\end{proposition}
\begin{proof}
Let $C = \{[f^d(\vec{a})] \mid d \in \gf{p}\}$ denote the set of
$\cong$-equivalence classes of $f^d(\vec{a})$.  By the fact that $f$
is a partial isomorphism over $[G]$, we have an action of the group
$\gf{p}$ on $C$ given by $[f^d(\vec{a})]^i = [f^{d+1}(\vec{a})]$ for $i
\in \gf{p}$.  Since $p$ is prime, $\gf{p}$ has no non-trivial
subgroups and so the kernel of this action is either trivial, in which
case $f^{d_1}(\vec{a}) \not\cong f^{d_2}(\vec{a})$ for all distinct
$d_1$ and $d_2$ or the kernel of the action is all of $\gf{p}$, in
which case $f^d(\vec{a}) \cong \vec{a}$ for all $d$.
\end{proof}

\begin{lemma}
\label{isomorph}
For $X \subseteq V$ and $\rho:X \to \gf{p}$, if $\rho$ induces a liftable isomorphism over $G$, then $\rho$ induces a partial isomorphism over $[G]$.
\end{lemma}
\begin{proof}
Let $f:= \bij(\rho)$, and let $S=\Dcal(X)^c\cap W$. We show that $f$ preserves the relations $[R]\cap S^2$ where $R$ is either one of $F$ or $\sim$.  To this end, let $(\vec{a}_1,\vec{a}_2) \in [R]\cap S^2$ and we show $(f(\vec{a}_1),f(\vec{a}_2))$ is in $[R]\cap S^2$ as well.  First we fix $(\vec{b}_1,\vec{b}_2) \in R$ such that $\vec{b_i} \cong \vec{a}_i$ for $i \in \{1,2\}$, which must exist by our choice of $\vec{a}_1,\vec{a}_2$.  For each $i$ there is then a sequence $\vec{d}_{i,1} (= \vec{a}_i),\vec{d}_{i,2},\ldots,\vec{d}_{i,s_i} (= \vec{b}_i)$ such that $\vec{d}_{i,j} \sim \vec{d}_{i,j+1}$ or $\vec{d}_{i,j+1} \sim \vec{d}_{i,j}$.  Let $Y_j$ be the set $V(\tilde{d}_{1,j})$ if $1 \leq j \leq s_1$ or $V(\tilde{d}_{2,s_1+s_2-j})$ if $s_1 < j  \leq s_1 + s_2$.  In other words, the sequence of $Y_j$ forms the vertex sets of the symmetric $\sim$-path from $\vec{a}_1$ to $\vec{b}_1$, followed by the symmetric $\sim$-path from $\vec{b}_2$ to $\vec{a}_2$.
Note that $|Y_j| \leq c$ for all $j$.  By the assumption that $\rho$ is a liftable isomorphism, there is a sequence $\sigma_j: Y_j\to \gf{n}$  in $\lift{\rho}{Y_1,Y_2,\ldots,Y_s}$ and $f_j := \bij(\rho\cup\sigma_j\cup\sigma_{j+1})$ defines a partial isomorphism over $G$ for all $j<s_1+s_2$.  This means that for $i \leq k$ and $j < s_1+s_2$ with $j\neq s_1$, $f_j(\vec{c}_{i,j}) \sim f_j(\vec{c}_{i,j+1})$ or $f_j(\vec{c}_{i,j+1}) \sim f_j(\vec{c}_{i,j})$, and $(f_{s_1}(\vec{b}_1),f_{s_1+1}(\vec{b}_2)) \in R$. Thus, we conclude that $(f(\vec{a}_1),f(\vec{a}_2)) \in [R]$, as required. Since all of the above holds for $f^{-1}$ as well,  it follows immediately that $f$ is a well defined partial isomorphism over $[G]$.
\end{proof}

For set $X\subseteq V$, function $\rho:X\to\gf{p}$, and node $u \in X$, we say $\rho$ is \emph{zero bar $u$} if $X\setminus\{u\}\subseteq \Zcal(\rho)$. We further say $\rho$ is \emph{a spike at $u$}, if $\rho$ is zero bar $u$ and $\rho(u) \neq 0$.

For $\vec{a} \in W$, $M$ the frontier of $\cl{V(\tilde{a})}$, and $u \in M\setminus \beta$, we say $u$ is \emph{free} in $[\vec{a}]$ if for all functions $\rho:\beta\cup M\to \gf{p}$ that are zero bar $u$, we have that $\rho$ is consistent and $[\bij(\cl{\rho}|_{V(\tilde{a})})](\vec{a}) \cong \vec{a}$. We say $u$ is \emph{bounded} in $[\vec{a}]$ otherwise. Note that any two functions $\rho,\sigma: \beta\cup M\to\gf{p}$ that are spikes at $u$, defining $f := \bij(\cl{\rho})$ and $g:= \bij(\cl{\sigma})$, we have that $g = f^d$ for some $d$. Thus, by Proposition~\ref{degree}, if $\rho$ is consistent and induces a partial isomorphism over $[G]$, then $u$ is free in $[\vec{a}]$ if and only if $f(\vec{a}) \cong \vec{a}$.

Say we have a sequence $\alpha$ of nodes $\vec{a}_1,\vec{a}_2,\ldots,\vec{a}_t$ in $W$, some $i\leq t$, node $u$ bounded in $[\vec{a}_i]$, and $M$ the frontier of $\cl{V(\tilde{a}_i)}$. We say $\alpha$ \emph{strongly bounds} $u$ at $i$ if for consistent $\rho:\beta\cup M\to\gf{p}$ that are zero bar $u$, all $j \in (i, t]$, and $\sigma:V(\tilde{a}_j)\to\gf{p}$ the function which takes all elements to zero, we have that either $\rho\cup\sigma$ is inconsistent or $\cl{\rho\cup\sigma}$ induces a liftable isomorphism over $[G]$. Finally, we say $\alpha$ is a \emph{strongly-bounding path} if both of the following are satisfied \begin{enumerate}
    \item $(\vec{a}_{i},\vec{a}_{i+1}) \in [F]$ for $i < t$ (i.e. the sequence is a directed path); and
    \item for all $i \leq t$, and node $u$ bounded in $[\vec{a}_i]$, $\alpha$ strongly bounds $u$ at $i$
\end{enumerate}The following is a key combinatorial property of strongly-bounding paths. 

\begin{proposition} \label{ihestablish} For a strongly-bounding path $\vec{a}_1,\vec{a}_2,\ldots,\vec{a}_s$ in $W$, $M$ the frontier of $\cl{V(\tilde{a}_1)}$, and node $u$ bounded in $[\vec{a}_1]$, if $u$ is not contained in $\cl{(M\setminus\{u\})\cup\beta \cup V(\tilde{a}_s)}$ then there is some $i>1$ such that $|[F]\vec{a}_i| \geq p$.
\end{proposition}
\begin{proof}
  Take $\rho: M\cup \beta \cup V(\tilde{a}_s)\to\gf{p}$ that is a spike at $u$. By our assumption that $u$ is not contained in $\cl{M\setminus\{u\}\cup \beta \cup V(\tilde{a}_s)}$, it follows that $\rho$ is consistent. Thus $\cl{\rho}$ induces a liftable isomorphism over $[G]$, by the definition of a strongly-bounding path.
  Next take the sequence $Y_i = V(\tilde{a}_i)$ for $i \leq s$ (note that $|Y_i| \leq c$).
  Since $\rho$ induces a liftable isomorphism over $[G]$, we know there exists a sequence $\sigma_i:Y_i\to\gf{p}$ in $\lift{\rho}{Y_1,Y_2,\ldots,Y_s}$ such that each $f_i := \bij(\sigma_i)$ is a partial isomorphism over $[G]$.
  Let $j \leq s$ be the least such that $f_j(\vec{a}_j) \cong \vec{a}_j$.  We know that such a $j$ exists because $\sigma_s$ is simply $\rho|_{V(\tilde{a}_s)}$, which is the constant zero function, and thus $f_s(\vec{a}_s) = \vec{a} \cong \vec{a}$.
It must also be the case that $j > 1$ because $\sigma_1$ is $\rho|_{V(\tilde{a}_1)}$, which induces a partial isomorphism over $[G]$ by  assumption, and thus since $u$ is bounded in $[\vec{a}_1]$ we have $f_1(\vec{a}_1) \not\cong \vec{a}_1$.  So, since $j-1 \geq 1$, we note by Proposition~\ref{degree}  that for all $d_1,d_2 \in \gf{p}$ with $d_1 \neq d_2$, $f^{d_1}_{j-1}(\vec{a}_{j-1}) \not\cong  f^{d_2}_{j-1}(\vec{a}_{j-1})$.  By  a similar argument, one also sees that for all $d_1,d_2 \in \gf{p}$, $f^{d_1}_{j}(\vec{a}_{j}) \cong  f^{d_2}_{j}(\vec{a}_{j})$.  Moreover, since $(\vec{a}_{j-1},\vec{a}_{j}) \in [F]$ and $\rho\cup\sigma_{j-1}\cup\sigma_{j}$ is consistent by Proposition~\ref{prop:sequence}, we can further conclude that for $d \in \gf{p}$, $(f^d(\vec{a}_{j-1}),f^d(\vec{a}_{j}))$ is also in $[F]$.  Together these gives us that $|[F]\vec{a}_j|\geq p$.
\end{proof}

The following technical lemma provides a useful bridge between the combinatorial properties established so far and the the graph move in the $\LREC_=$ game.

\begin{lemma} \label{final} Say we have $k \geq 2(c + |\beta|)$, a strongly-bounding path $\vec{a}_1,\vec{a}_2,\ldots,\vec{a}_t \in W$, consistent $\rho_i,\sigma_i:Y_i\to\gf{p}$ where $\sigma_i(u) = 0$ if $u$ is free in $[\vec{a}_i]$ and $\sigma_i(u) = \rho_i(u)$ otherwise, and set $A:= \{\vec{a}_i\mid |[F]\vec{a}_i| \geq p\}$.
Suppose for $i < t$ there is some $\eta \in \lift{\sigma_i}{Y_{i+1}}$ with $\rho_{i+1} = \eta|_{Y_{i+1}}$, then for all $r \geq 0$, if $\nh{\rho_1} - \nh{\rho_t} \geq k^3r$ then $|A| \geq r$.
\end{lemma}
\begin{proof}
For $i \leq t$, let $A_i$ denote the set $\{\vec{a}_j \in A \mid j \leq i\}$ and fix $h \geq k^3$.  We prove the following stronger claim by induction on $\ell$, which immediately implies the lemma.

\begin{claim*}
For $\ell\leq t$ and node $u$ in $Y_\ell\setminus\Zcal(\rho_\ell)$ the following holds
\begin{enumerate}
    \item There exists a sequence $\mathcal{K} := (v_i)_{0\leq i \leq s}$ of nodes, and sequence of natural numbers $\ell < m_1 < \cdots < m_s = 1$ such that the following is true:
\begin{itemize}
    \item $v_0= u$ and $v_s \in Y_1$ is an ancestor of $u$.
    \item For every $j>0$, $v_j$ is in the path from $v_0$ to $v_s$ and is bounded in $[\vec{a}_{m_j}]$
    \item For every $0 \leq j < s$, $3k+1 \leq $ height$(v_{j+1})-$height$(v_j) \leq 4k$.
\end{itemize}
\item For $r\geq 0$, if $\nh{\rho_1} - \hgt{u} \geq hr$ then $|A_\ell| \geq r$
\end{enumerate}

\end{claim*}
Part 1 of the claim is intended only to aid us in our proof of part 2, which gives the lemma. We proceed by induction, so consider the case where $\ell=1$. Then trivially we can take $\mathcal{K}$ to be $\{u\}$. For the second part of the claim, we note that $\nh{\rho_1} - \nh{\rho_1} = 0$ and thus the only interesting case is when $r$ is equal to $0$ by our choice of $h$. Trivially, $|A_1| \geq 0$, so the base case holds.

For the induction step let $\ell > 1$. Then we can see by our choice of the sequence $\rho_i,\sigma_i$ and Proposition~\ref{prop:sequence} that there must be some some ancestor $y$ of $u$ with $3k+1 \leq $ height$(y)-$height$(u) \leq 4k$ and $\sigma_i(y) \neq 0$ for some $i<\ell$. In particular this means $y$ is bounded in $[\vec{a}_i]$ by our choice of $\sigma_i$. Moreover, there must be sequences $\mathcal{K}'$ and $m_j'$ that satisfy the induction hypothesis on $i$ and $y$. We then let $\mathcal{K}$ be the sequence given by prepending $u$ to $\mathcal{K}'$, $m_1$ be $i$, and $m_j$ be $m_{j-1}'$ for $j>1$. It is easy to verify this choice of $\mathcal{K}$ and $m_i$ satisfy the required conditions, so we have shown part 1 of the claim.

For the second part take arbitrary $r \geq 0$ and assume $\nh{\rho_1} - \hgt{u} \geq hr$. If $r=0$ then we are trivially done, so assume $r>1$. Let $j\leq s$ be the largest index for which $\nh{\rho_1}-\hgt{v_j} > h(r-1)$ where $v_j \in \Kcal$. We begin by showing that $\nh{\rho_1}-\hgt{v_j} \leq h(r-1) + 4k$. Indeed, if $m_j$ is equal to $1$ this is trivial. Otherwise, it must be true that $j$ is less than $s$, since $m_s$ is equal to $1$ by construction.  Thus $\hgt{v_{j+1}}-\hgt{v_j} \leq 4k$, which  implies that $\nh{\rho_1}-\hgt{v_j} \leq h(r-1) + 4k$ as required, since our choice of $j$ implies that $\nh{\rho_1}-\hgt{v_{j+1}} \leq h(r-1)$.  From this we further note that 
$$j \geq (\hgt{v_j} - \hgt{u})/4k \geq (h-4k)/4k$$
by our choice of $\Kcal$.  Hence $j \geq k$ by our choice of $h$.

We now put this together to prove that $|A_\ell| > |A_{m_j}|$from which the claim follows using the inductive hypothesis stating $|A_{m_j}| \geq r-1$.  To derive a contradiction we assume that $|A_\ell| = |A_{m_j}|$. For $i \leq j$, let $M_i$ be the frontier of $\cl{V(\tilde{a}_{m_i})}$ and let $X_i$ be $M_i\setminus\{v_i\} \cup \beta \cup V(\tilde{a}_\ell)$. By our assumption that $|A_\ell| = |A_{m_j}|$ and proposition \ref{ihestablish}, we know that $v_i$ must be contained in $\cl{X_i}$ for all $i\leq j$, as otherwise there would be some $\ell' \in (m_j,m_i]$ with $|[F]\vec{a}_{\ell'}| \geq p$, contradicting the assumption. We show by induction on $i$ that this implies that $$ |T(c_i)\cap (\beta\cup V(\tilde{a}_\ell))| \geq i$$ where $c_i$ is the ancestor of $v_i$ with $\hgt{c_i} - \hgt{v_i} = 2k$. From here we can conclude that $|T(c_j)\cap (\beta\cup V(\tilde{a}_\ell))| \geq k$, which is a contradiction, since $|\beta\cup V(\tilde{a}_\ell)|< k$.

The base case where $i = 0$ is trivial. For the induction step the goal is to show that $(T(c_i)\setminus T(c_{i-1}))\cap (\beta\cup V(\tilde{a}_\ell))$ is non empty, as that implies:
$$|T(c_i)\cap (\beta\cup V(\tilde{a}_\ell))| \geq 1 + |T(c_{i-1})\cap (\beta\cup V(\tilde{a}_\ell))| \geq 1 + (i-1) = i.$$

Let $Y\subseteq X_i\cup\{v_i\}$ be a minimal dependent set containing $v_i$. We know such a $Y$ exists because $v_i$ is in $\cl{X_i}$. Note then that $\hgt{Y}$ can be no greater than $|X_i| + 1 \leq k$ by Proposition~\ref{prop:sequence} and thus, $\mxh{Y} \leq \hgt{v_i}+k$ and $\mnh{Y} \geq \hgt{v_i} - k$. In particular this means that $Y \subseteq T(c_i)\setminus T(c_{i-1})$. Moreover, $Y$ cannot be a subset of $M_i$ since $M_i$ is independent, so it must be the case that $Y\cap (\beta\cup V(\tilde{a}_\ell))$ is non empty.  We conclude that $(T(c_i)\setminus T(c_{i-1}))\cap (\beta\cup V(\tilde{a}_\ell))$ is non empty and we are done.

\end{proof}

\begin{lemma}\label{lem:main-game}
  Let $n$ be an integer and $p$ a prime such that $p \geq n \geq
  k^3(3q+1)(k+1)$.  For $T$ a complete binary tree of height $n$ and
  $L$ the set of its leaves, let $\sigma: L \to \gf{p}$ be any function.
  Then, for any $t, t' \in \gf{p}$, Duplicator has a winning strategy
  in the $k$-step $q$-degree $\LREC_=$ game played on $\Pcal(n,p,\sigma,t)$ and
  $\Pcal(n,p,\sigma,t')$.
\end{lemma}
\begin{proof}
  We prove by induction on $m$ that for $X\subseteq V$ with $|X| \leq k-m$, and $\rho:X \to \gf{p}$ with $\nh{\rho} \geq k^3(3q+1)m$, the duplicator has a winning strategy on the $k$-pebble, $q$-degree game over structures $\Afrak,\Bfrak$ and function $\rho$.  
We notice that trivially, our starting pebble configuration $\rho_{init}$ is a special case of the above, so the induction gives us the result.

For $m = 0$ this is trivial, so assume $m > 0$. Let $\beta$ be $\dom(\rho)$ and $f$ be $\bij(\rho)$. If the spoiler plays an extension move then  Duplicator picks the function $\sigma:V\to \gf{p}$ which extends $\cl{\rho}$ by taking all elements in $V\setminus\cl{\beta}$ to 0. By Proposition~\ref{prop:sequence}, for any $a \in V$, $\sigma|_{X\cup\{a\}}$ is consistent and has null-height at least $k^3(3q+1)m - k > k^3(3q+1)(m-1)$. Thus, by the inductive hypothesis,  Duplicator has a winning strategy when  Spoiler plays the extension move.

Suppose that  Spoiler plays a graph move instead. Then it chooses some $c \leq (k-|\beta|)/2$, $d \leq q$, vectors of literals $\vec{x},\vec{y}$ of length $c$, and $\LREC_=$ queries  $\varphi_F(\vec{x},\vec{y})$, and $\varphi_\sim(\vec{x},\vec{y})$ with iteration-degree at most $q$, and rank at most $k-|\beta|-2c$.  Le the tuple $\vec{b} = (b_1,b_2,\ldots,b_{|\beta|})$ enumerate $\beta$. Then, let $W$ be $\Dcal(V)$, $G_\Afrak = (W, F_\Afrak, \sim_\Afrak)$ be the semi-graph defined in  $(\Afrak,\vec{b})$ by the interpretation of $(\varphi_F,\varphi_\sim)$, and $G_\Bfrak = (W, F_\Bfrak, \sim_\Bfrak)$ be the corresponding semi-graph obtained from $(\Bfrak, f(\vec{b}))$.  Spoiler chooses some starting node $\vec{a}_0 \in W\cap \Dcal(\beta)^c$, and starting counter $\ell_0 \leq |\Dcal(V)|^d$. At each step,  Duplicator's strategy is the following. Given $\rho_i$, it fixes $\sigma_i:V(\tilde{a}_i)\to\gf{p}$ with $\sigma_i(u) = 0$ if $u$ is free in $[\vec{a}_i]$ and $\sigma_i(u) = \rho_i(u)$ otherwise. For all $Y \subseteq V$ with $|Y| \leq c$, it then picks $\eta_Y$ to be in $\lift{\sigma_i}{Y}|_{Y}$. 

 For any $G_1,G_2 \in \{G_\Afrak,G_\Bfrak\}$, $X \subseteq V$ with $\beta \subseteq X$ and consistent $\mu:X\to\gf{p}$, we note that by the induction hypothesis on $m-1$, if $\nh{\mu}$ is at least $k^3(3q+1)(m-1)$ then $\mu$ induces a partial isomorphism between $G_1$ and $G_2$. Thus,  we have that if $\nh{\mu}$ is at least $k^3(3q+1)(m-1) + 2k$ then by Proposition~\ref{prop:sequence}, $\mu$ induces a liftable isomorphism between $G_1$ and $G_2$ and consequently by Lemma~\ref{isomorph}, $\mu$ induces a partial isomorphism between $[G_1]$ and $[G_2]$. Thus, if $\nh{\mu}$ is at least $k^3(3q+1)(m-1) + 4k$ then again by Proposition~\ref{prop:sequence}, $\mu$ induces a liftable isomorphism between $[G_1]$ and $[G_2]$.

We show that when  Duplicator plays following the above strategy, $\nh{\rho_i}$ is greater than $k^3(3q+1)(m-1) + 4k$ for all $i$.  This has a series of consequences that ultimately imply that  Duplicator has a winning strategy, and we spell these out first. Let $M_i$ be the frontier of $\cl{V(\tilde{a}_i)}$ for any $i$. For any $u \in M_i\setminus \beta$ which is free in $[\vec{a}_i]$, let $\mu:\beta\cup M_i \to \gf{p}$ be the function that is zero bar $u$ and takes $u$ to $\sigma_i(u) - \rho_i(u)$, and $g_u$ be the partial isomorphism induced by $\cl{\mu}|_{\beta\cup V(\tilde{a}_i)}$. Note that $\nh{\cl{\mu}}$ is no smaller than $\nh{\rho_i}$, and thus $g_u$ defines a partial isomorphism over $[G_\Afrak]$. Then, let $g$ be the unique composition of all the $g_u$, where $u$ is free in $[\vec{a}_i]$. It is easy to see that $g$ is equal to $\bij(\sigma_i-\rho_i)$, and since each $g_u$ defines a partial isomorphism over $[G_\Afrak]$, by repeated application of Proposition~\ref{degree} we can show that $g(\vec{a}_i) \cong \vec{a}_i$, and thus $\sigma_i$ is a valid choice by Duplicator. 

Moreover, fix some $Y \subseteq V$ and let $g$ now be the function induced by $\sigma_i \cup \eta_Y$. Note again that $\nh{\sigma_i}$ is no smaller than $\nh{\rho_i}$, and thus $\nh{\sigma_i\cup \eta_Y}$ is no smaller than $k^3(3q+1)(m-1) + 2k$ when $|Y|$ is at most $c$ by Proposition~\ref{prop:sequence}, which implies $g$ defines a partial isomorphism between $[G_\Afrak]$ and $[G_\Bfrak]$. Therefore, Duplicator's choices for $\eta_Y$ are valid. Finally, we can see by a similar argument that the in-degree of $[\vec{a}_i]$ in $[G_\Afrak]$ is the same as the in-degree of $[f_i(\vec{a}_i)]$ in $[G_\Bfrak]$. Thus we have verified that as long as $\nh{\rho_i}$ is greater than $k^3(3q+1)(m-1) + 4k$ for all $i$, the graph move only ends when the spoiler decides to reset $\rho$ to $\rho_s$ for some $s$. By our assumption on $\nh{\rho_s}$, and the inductive hypothesis on $m-1$, Duplicator has a winning strategy when such a decision is made by  Spoiler.

Therefore, if we can prove that $\nh{\rho_i}$ is greater than $h(3q+1)(m-1) + 4k$ for all $i$ then we are done. To derive a contradiction, assume there is some $s$ such that $\nh{\rho_s} \leq k^3(3q+1)(m-1) + 4k$. Then, let $t \leq s$ be the smallest such that $\nh{\rho_t} \leq k^3(3q+1)(m-1) + k^3$ and take $A:= \{\vec{a}_i\mid |[F]\vec{a}_i| \geq p \text{ AND } 0 < i \leq t\}$. By our choice of $t$, we know that $\nh{\sigma_{t-1}} \geq \nh{\rho_{t-1}} \geq k^3(3q+1)(m-1) + k^3$ and thus for $k$ large enough we have $\nh{\rho_t} > k^3(3q+1)(m-1) + 4k$ by Proposition~\ref{prop:sequence}. It is easy to verify that $\vec{a}_0,\vec{a}_1,\ldots,\vec{a}_t$ is a strongly-bounding path, so we invoke Lemma~\ref{final} to get that $|A| \geq 3q$ (since $\nh{\rho_1}-\nh{\rho_t} \geq 3qk^3$ by choice of $t$). Note that $|V\times\gf{p}|$ is no greater than $p^2$, and thus $\ell_0$ can be at most $p^{2q}$. However, this implies that $\ell_{t-1}$ can be at most $\ell_0/p^{|A|-1} < 1$. Hence, we have a contradiction: as $\ell_{t-1}$ must be 0, the graph move must terminate prior to reaching step $t$. This completes the proof.
\end{proof}

From this our main theorem follows immediately.
\begin{theorem}
\label{main}
\PSP is not expressible in $\LREC_=$
\end{theorem}
\begin{proof}
 Let $\tau$ be the vocabulary of \PSP instances and suppose for
 contradiction that there is a sentence $\varphi$ of $\LREC_=[\tau]$
 that defines the positive instances.  Fix $k$ to be greater than the
 rank of $\phi$ and $q$ to be grater than $\dg{\varphi}$.  Take $T$ to
 be a complete binary tree of height $n \geq k^3(3q+1)(k+1)$ and $p$
 to be a prime with $p \geq n$.  Fix a function $\sigma$ from the
 leaves $L$ of $T$ to $\gf{p}$ and let $\Afrak$ be the structure
 $\Pcal(n,p,\sigma,t)$ where $t = \sum_{l \in L} \sigma(l)$ and
 $\Bfrak$ to be the structure  $\Pcal(n,p,\sigma,t')$ where $t' =
 t+1$.  Then, by construction $\Afrak$ is a positive instance of
 $\PSP$ so $\Afrak \models \phi$ and $\Bfrak$ is a negative instance,
 so $\Bfrak \not\models \phi$.  But, by Lemma~\ref{lem:main-game}, the
 two structures are indistinguishable by $\phi$ and we have a contradiction.

\end{proof}

%% file: conclusion.tex
Over ordered structures, the logic $\LREC_=$ captures \LS, while $\FPC$ captures \PT. So, it is of value to investigate the relationship between the two logics. It was shown in \cite{Grohe_2013} that $\LREC_=$ is contained in $\FPC$, and the main result of this paper shows the containment is proper. Our proof here  uses novel techniques that provide insight on the expressive power of $\LREC_=$ and some of the combinatorial properties of \PSP that make it difficult to be solved in \LS. 

These techniques can also be used to find further results. Firstly, the Ehrenfeucht–Fraïssé game provides a new tool that can be used and adapted for other inexpressibility results of the logic. Secondly, we note that the combinatorial results in Section~\ref{sec:extender} use the characterization of a tree as a matroid. The results can further be generalized to a larger class of finite matroids, which would include constructing a more general definition of height for arbitrary matroids. It is not clear whether it could be extended to all classes of finite matroids however, as the argument relies on the fact that the independent sets of the matroid is described concisely in the form of the edge relation of the tree, and it is not true that all classes of finite matroids can be described by a finite vocabulary. The results of Section~\ref{sec:result} can be generalized to larger classes of structures as well. This may for example include considering other classes of matroids as the base set following from the discussion above. Or, it may include constructing interpretations of structures other than semi-graphs, as all that is required for the result is a binary relation $\sim$, which is used to define a quotient structure of the interpretation through its transitive reflexive symmetric closure.

The Path Systems Problem is not known to be in $\LS$, and thus the question of whether $\LREC_=$ is equal to $\FPC \cap \LS$ remains open. A result that they are equal would imply that $\LS \neq \PT$, since $\PSP$ is in $\LFP$ and hence in $\FPC$. It also remains open for many classes of structures $\Ccal$ that are known to have logarithmic-space canonical labelling algorithms whether  $\LREC_=$ captures $\LS$ over $\Ccal$. Studying these classes may prove fruitful, as it is a direction for solving the  $\LREC_=$ versus $\FPC\cap\LS$ problem. 

As mentioned before, it remains open whether $\PSP$ is expressible in some logics stronger than $\LREC_=$, including the logic $\CL$.  Of course, one can also ask whether $\PSP$ is in $\LS$ and this is equivalent to asking if $\LS$ is equal to $\PT$, but the question of whether $\PSP$ is in $\CL$ does not seem to have been investigated. Another related long-standing open problem is to find a logic that captures \LS over all structures.